\documentclass[10pt,a4paper]{article}
\usepackage{amsmath,amsthm,enumerate}
\usepackage[hmargin=3cm,vmargin=3cm]{geometry}

\usepackage[utf8]{inputenc}
\usepackage[T1]{fontenc}
\usepackage[english]{babel}
\usepackage{hyperref}
\usepackage{amsmath}
\usepackage{tikz}
\usepackage[color=green!30]{todonotes}
\usepackage{xspace}

\tikzset{
blackvertex/.style={circle, draw=black!100,fill=black!100,thick, inner sep=0pt, minimum size=2mm}
}

\tikzset{
whitevertex/.style={circle, draw=black!100,thick, inner sep=0pt, minimum size=2mm}
}

\newcommand{\M}{EM}

\DeclareMathOperator{\EM}{dem}

\DeclareMathOperator{\SEM}{sdem}

\DeclareMathOperator{\VC}{vc}
\DeclareMathOperator{\FES}{fes}
\DeclareMathOperator{\FVS}{fvs}
\DeclareMathOperator{\arb}{arb}

\newcommand{\decisionpb}[4]{
\begin{center}
        \noindent\framebox{\begin{minipage}{#4\textwidth}
                #1\\
                Instance: #2\\ 
                Question: #3
        \end{minipage}}
\end{center}
}

\newcommand{\optimpb}[5]{
\begin{center}
        \noindent\framebox{\begin{minipage}{#4\textwidth}
                #1\\
                Instance: #2\\ 
                Task: #3
        \end{minipage}}
\end{center}
}

\newcommand{\EMdec}{\textsc{Distance-Edge-Monitoring Set}\xspace}
\newcommand{\EMopt}{\textsc{Min Distance-Edge-Monitoring Set}\xspace}
\newcommand{\SCdec}{\textsc{Set Cover}\xspace}
\newcommand{\SCopt}{\textsc{Min Set Cover}\xspace}

\newtheorem{theorem}{Theorem}

\newtheorem{proposition}[theorem]{Proposition}
\newtheorem{lemma}[theorem]{Lemma}
\newtheorem{observation}[theorem]{Observation}
\newtheorem{corollary}[theorem]{Corollary}

\newtheorem{definition}[theorem]{Definition}

\begin{document}
\sloppy

\title{Monitoring the edges of a graph using distances\thanks{This paper is dedicated to the memory of Mirka Miller, who sadly passed away in January 2016. This research was started when Mirka Miller and Joe Ryan visited Ralf Klasing at the LaBRI, University of Bordeaux, in April 2015.}~\thanks{A preliminary version of this paper appeared as~\cite{CALDAM}.}~\thanks{The authors acknowledge the financial support from the ANR project HOSIGRA~(ANR-17-CE40-0022), the IFCAM project ``Applications of graph homomorphisms''~(MA/IFCAM/18/39), the Programme IdEx Bordeaux -- SysNum (ANR-10-IDEX-03-02) and the French government IDEX-ISITE initiative 16-IDEX-0001 (CAP 20-25).}}

\author{Florent Foucaud\footnote{\noindent LIMOS, CNRS UMR 6158, Universit\'e Clermont Auvergne, Aubi\`ere, France.}~~\footnote{\noindent Universit\'e de Bordeaux, Bordeaux INP, CNRS, LaBRI, UMR 5800, Talence, France.}~~\footnote{Univ. Orléans, INSA Centre Val de Loire, LIFO EA 4022, F-45067 Orléans Cedex 2, France.}
    \and Shih-Shun Kao\footnotemark[5]~~\footnote{Department of Computer Science and Information Engineering, National Cheng Kung University, Tainan 701, Taiwan.}
  \and Ralf Klasing\footnotemark[5]
  \and Mirka Miller\footnote{School of Mathematical and Physical Sciences, The University of Newcastle, Australia and Department of Mathematics, University of West Bohemia, Czech Republic.}
  \and Joe Ryan\footnote{School of Electrical Engineering and Computing, The University of Newcastle, Newcastle, Australia.}
}

\maketitle

\begin{abstract}
  We introduce a new graph-theoretic concept in the area of network monitoring. A set $M$ of vertices of a graph $G$ is a \emph{distance-edge-monitoring set} if for every edge $e$ of $G$, there is a vertex $x$ of $M$ and a vertex $y$ of $G$ such that $e$ belongs to all shortest paths between $x$ and $y$. We denote by $\EM(G)$ the smallest size of such a set in $G$. The vertices of $M$ represent distance probes in a network modeled by $G$; when the edge $e$ fails, the distance from $x$ to $y$ increases, and thus we are able to detect the failure. It turns out that not only we can detect it, but we can even correctly locate the failing edge.

  In this paper, we initiate the study of this new concept. We show that for a nontrivial connected graph $G$ of order $n$, $1\leq\EM(G)\leq n-1$ with $\EM(G)=1$ if and only if $G$ is a tree, and $\EM(G)=n-1$ if and only if it is a complete graph. We compute the exact value of $\EM$ for grids, hypercubes, and complete bipartite graphs.

  Then, we relate $\EM$ to other standard graph parameters. We show that $\EM(G)$ is lower-bounded by the arboricity of the graph, and upper-bounded by its vertex cover number. It is also upper-bounded by twice its feedback edge set number.
Moreover, we characterize connected graphs $G$ with $\EM(G)=2$.

  Then, we show that determining $\EM(G)$ for an input graph $G$ is an NP-complete problem, even for apex graphs. There exists a polynomial-time logarithmic-factor approximation algorithm, however it is NP-hard to compute an asymptotically better approximation, even for bipartite graphs of small diameter and for bipartite subcubic graphs. For such instances, the problem is also unlikey to be fixed parameter tractable when parameterized by the solution size.
\vspace*{9mm}

\end{abstract}

\section{Introduction}

The aim of this paper is to introduce a new concept of network monitoring using distance probes, called \emph{distance-edge-monitoring}. Our networks are naturally modeled by finite undirected simple connected graphs, whose vertices represent computers and whose edges represent connections between them. We wish to be able to monitor the network in the sense that when a connection (an edge) fails, we can detect this failure. We will select a (hopefully) small set of vertices of the network, that will be called \emph{probes}. At any given moment, a probe of the network can measure its graph distance to any other vertex of the network. Our goal is that, whenever some edge of the network fails, one of the measured distances changes, and thus the probes are able to detect the failure of any edge.

Probes that measure distances in graphs are present in real-life networks, for instance this is useful in the fundamental task of \emph{routing}~\cite{DABVV06,GT00}. They are also frequently used for problems concerning \emph{network verification}~\cite{BBDGKP15,BEEHHMR06,BEMW10}.

We will now proceed with the formal definition of our main concept. In this paper, by \emph{graphs} we refer to connected simple graphs (without multiple edges and loops). A graph with loops or multiple edges is called a \emph{multigraph}.

We denote by $d_G(x,y)$ the distance between two vertices $x$ and $y$ in a graph $G$. 
For an edge $e$ of $G$, we denote by $G-e$ the graph obtained by deleting $e$ from $G$. 

\begin{definition}
  For a set $M$ of vertices and an edge $e$ of a graph $G$, let $P(M,e)$ be the set of pairs $(x,y)$ with $x$ a vertex of $M$ and $y$ a vertex of $V(G)$ such that $d_G(x,y)\neq d_{G-e}(x,y)$. In other words, $e$ belongs to all shortest paths between $x$ and $y$ in $G$.

  For a vertex $x$, let $\M(x)$ be the set of edges $e$ such that there exists a vertex $v$ in $G$ with $(x,v)\in P(\{x\},e)$. If $e\in \M(x)$, we say that $e$ is \emph{monitored} by $x$.

  A set $M$ of vertices of a graph $G$ is \emph{distance-edge-monitoring} if every edge $e$ of $G$ is monitored by some vertex of $M$, that is, the set $P(M,e)$ is nonempty. Equivalently, $\bigcup_{x\in M} \M(x)=E(G)$. 

  We denote by $\EM(G)$ the smallest size of a distance-edge-monitoring set of $G$.
\end{definition}

Note that $V(G)$ is always a distance-edge-monitoring set of $G$, so $\EM(G)$ is always well-defined.

Consider a graph $G$ modeling a network, and a set $M$ of vertices of $G$, on which we place probes that are able to measure their distances to all the other vertices. If $M$ is distance-edge-monitoring, if a failure occurs on any edge of the network (in the sense that the communication between its two endpoints is broken), then this failure is detected by the probes.

In fact, it turns out that not only the probes can \emph{detect} a failing edge, but they can also precisely \emph{locate} it (the proof of the following is delayed to Section~\ref{sec:prelim}).

\begin{proposition}\label{prop:location}
Let $M$ be a distance-edge-monitoring set of a graph $G$. Then, for any two distinct edges $e$ and $e'$ in $G$, we have $P(M,e)\neq P(M,e')$.
\end{proposition}

Thus, assume that we have placed probes on a distance-edge-monitoring set $M$ of a network $G$ and initially computed all the sets $P(M,e)$. In the case a unique edge of the network has failed, Proposition~\ref{prop:location} shows that by measuring the set of pairs $(x,y)$ with $x\in M$ and $y\in V(G)$ whose distance has changed, we know exactly which is the edge that has failed.

We define the decision and optimization problem associated to distance-edge-monitoring sets.

\decisionpb{\EMdec}{A graph $G$, an integer $k$.}{Do we have $\EM(G)\leq k$?}{0.9}

\optimpb{\EMopt}{A graph $G$.}{Build a smallest possible distance-edge-monitoring set of $G$.}{0.9}

\paragraph{Related notions.} A weaker model is studied in~\cite{BBDGKP15,BEEHHMR06} as a \emph{network discovery} problem, where we seek a set $S$ of vertices such that for each edge $e$, there exists a vertex $x$ of $S$ and a vertex $y$ of $G$ such that $e$ belongs to \emph{some} shortest path from $x$ to $y$.

A different (also weaker) model is the \emph{Link Monitoring problem} studied in~\cite{BR06}, in which one seeks to monitor the edges of a graph network by selecting vertices to act as probes. To each probe is assigned a routing tree (a DFS tree spanning the whole graph), and it is essentially required that each edge of the graph belongs to one of the trees.


Distance-edge-monitoring sets are also related to \emph{resolving sets} and \emph{edge-resolving sets}, that model sets of sensors that can measure the distance to all other vertices in a graph. A resolving set is a set $R$ of vertices such that for any two distinct vertices $x$ and $y$ in $G$, there is a vertex $r$ in $R$ such that $d_G(r,x)\neq d_G(r,y)$. The smallest size of a resolving set in $G$ is the \emph{metric dimension} of $G$~\cite{HM76,S75}. If instead, the set $R$ distinguishes the edges of $G$ (that is, for any pair $e,e'$ of edges of $G$, there is a vertex $r\in R$ with $d_G(x,e)\neq d_G(x,e')$, where for $e=uv$, $d_G(x,e)=\min\{d_G(x,u),d_G(x,v)\}$), we have an edge-resolving set~\cite{edgeMD}. 
A set that distinguishes all vertices and edges is a \emph{mixed resolving set}~\cite{mixedMD}.
Note that there is no relation, in general, between the (edge-)metric dimension of $G$, and $\EM(G)$, as shown by the examples of trees and grids: trees can have arbitrarily large metric dimension and edge-metric dimension~\cite{edgeMD}, but when $G$ is a tree $\EM(G)=1$ (Theorem~\ref{thm:trees}). Conversely, $\EM(G)$ can be arbitrarily large for grids (Theorem~\ref{thm:grids}), while their edge-metric dimension and metric dimension is $2$~\cite{edgeMD}.

Another related concept is the one of \emph{strong resolving sets} \cite{strong-resolving,ST04}: a set $R$ of vertices is strongly resolving if for any pair $x,y$ of vertices, there exists a vertex $z$ of $R$ such that either $x$ is on a shortest path from $z$ to $y$, or $y$ is on a shortest path from $z$ to $x$. It is related to distance-edge-monitoring sets in the following sense. Given a distance-edge-monitoring set $M$, for every pair $x,y$ of \emph{adjacent} vertices, there is a vertex $z$ of $M$ such that either $x$ is on \emph{every} shortest path from $z$ to $y$, or $y$ is on \emph{every} shortest path from $z$ to $x$.

The concept of an \emph{(strong) edge-geodetic set} is also related to ours. A set $S$ of vertices is an edge-geodetic set if for every edge $e$ of $G$, there are two vertices $x,y$ of $S$ such that $e$ is on some shortest path from $x$ to $y$. It is a strong edge-geodetic set if to every pair $x,y$ of $S$, we can assign a shortest $x-y$ path to $\{x,y\}$ such that every edge of $G$ belongs to one of these $|S|\choose 2$ assigned shortest paths~\cite{geodetic}.

A model related (only by name) is the one of \emph{edge monitoring sets}~\cite{edgemon-triangle}: a set $S$ is edge monitoring if for every edge $xy$ of $G$, $e$ belongs to a triangle $xyz$ with $z\in S$. But this is very far from our definition.

\paragraph{Our results.} We first derive a number of basic results about distance-edge-monitoring sets in Section~\ref{sec:prelim} (where we also give some useful definitions).

In Section~\ref{sec:bounds}, we study $\EM$ for some basic graph families (like trees and grids) and relate this parameter to other standard graph parameters such as arboricity $\arb$, vertex cover number $\VC$ and feedback edge set number $\FES$. We show that $\EM(G)=1$ if and only if $G$ is a tree. We show that for any graph $G$ of order $n$, $\EM(G)\geq\arb(G)$. Moreover, $\EM(G)\leq\VC(G)\leq n-1$ (with equality if and only if $G$ is complete). We show that for some families of graphs $G$, $\EM(G)=\VC(G)$, for instance this is the case for complete bipartite graphs and hypercubes. Then we show that $\EM(G)\leq 2\FES(G)-2$ when $\FES(G)\geq 3$ (when $\FES(G)\leq 2$, $\EM(G)\leq \FES(G)+1$).\footnote{This is an improvement over the results from the conference version~\cite{CALDAM}, where we proved $\EM(G)\leq 5\FES(G)-5$.}

In Section~\ref{sec:dem=2}, we characterize connected graphs $G$ with $\EM(G)=2$.

In Section~\ref{sec:complex}, we show that \EMdec is NP-complete, even for apex graphs (graphs obtained from a planar graph by adding an extra vertex). Then, we show that \EMopt can be approximated in polynomial time within a factor of $\ln(|E(G)|+1)$ by a reduction to the set cover problem. Finally, we show that no essentially better ratio can be obtained (unless P$=$NP), even for graphs that are of diameter~$4$, bipartite and of diameter~$6$, or bipartite and of maximum degree~$3$. For the same restrictions, the problem is unlikely to be fixed parameter tractable when parameterized by the solution size. These hardness results are obtained by reductions from the \textsc{Set Cover} problem.

We conclude our paper in Section~\ref{sec:conclu}.

\section{Preliminaries}\label{sec:prelim}



We now give some useful lemmas about basic properties of distance-edge-monitoring sets. We start with the proof of Proposition~\ref{prop:location}.

\begin{proof}[Proof of Proposition~\ref{prop:location}]
Suppose by contradiction that there are two distinct edges $e$ and $e'$ with $P(M,e)=P(M,e')$. Since by definition, $P(M,e)\neq\emptyset$, we have $(x,v)\in P(M,e)$ (and thus $(x,v)\in P(M,e')$), where $x\in M$ and $v\in V(G)$. Thus, all shortest paths between $x$ and $v$ contain both $e$ and $e'$. Assume that on one of these shortest paths, $e$ is closer to $x$ than $e'$, and let $w$ be the endpoint of $e$ that is closest to $v$. Then, we have $(x,w)\in P(M,e)$ but $(x,w)\notin P(M,e')$: a contradiction.
\end{proof}

An edge $e$ in a graph $G$ is a \emph{bridge} if $G-e$ has more connected components than $G$. We will now show that bridges are very easy to monitor.

\begin{lemma}\label{lemma:bridges}
Let $G$ be a connected graph and let $e$ be a bridge of $G$. For any vertex $x$ of $G$, we have $e\in \M(x)$.
\end{lemma}
\begin{proof}
Assume that $e=uv$. Since $e$ is a bridge, we have $d_G(x,u)\neq d_G(x,v)$. If $d_G(x,u)<d_G(x,v)$, then $(x,v)\in P(\{x\},e)$. Otherwise, we have $(x,u)\in P(\{x\},e)$. In both cases, $e\in \M(x)$.
\end{proof}

We now introduce the following terminology from~\cite{ELW12j}. In a graph, a vertex is a \emph{core vertex} if it has degree at least~$3$. A path with all internal vertices of degree~$2$ and whose end-vertices are core vertices is called a \emph{core path} (note that we allow the two end-vertices to be equal, but all other vertices must be distinct). A core path that is a cycle (that is, both end-vertices are equal) is a \emph{core cycle}. The \emph{base graph}  of a graph $G$ is the graph obtained from $G$ by iteratively removing vertices of degree~$1$ (thus, the base graph of a forest is the empty graph).


Lemma~\ref{lemma:bridges} implies the following.

\begin{observation}\label{obs:EM2}
Let $G$ be a graph, and $G_b$ be its base graph. Then, $\EM(G)=\EM(G_b)$.
\end{observation}

%
%

Given a vertex $x$ of a graph $G$ and an integer $i$, we let $L_i(x)$ denote the set of vertices at distance $i$ of $x$ in $G$. 

\begin{lemma}\label{lemma:M(x)}
Let $x$ be a vertex of a connected graph $G$. Then, an edge $uv$ belongs to $\M(x)$ if and only if $u\in L_i(x)$ and $v$ is the only neighbour of $u$ in $L_{i-1}(x)$, for some integer $i$.
\end{lemma}
\begin{proof}
  Let $uv\in \M(x)$. Then, there exists a vertex $y$ such that all shortest paths from $y$ to $x$ go through $uv$. Thus, one of $u$ and $v$ (say, $u$) is in a set $L_i(x)$ and the other ($v$) is in a set $L_{i-1}(x)$, for some positive integer $i$. Moreover, $v$ must be the only neighbour of $u$ in $L_{i-1}(x)$, since otherwise there would be a shortest path from $y$ to $x$ going through $u$ but avoiding $uv$.

Conversely, if $u$ is a vertex in $L_i(x)$ with a unique neighbour $v$ in $L_{i-1}(x)$, then the edge $uv$ belongs to $\M(x)$ since all shortest paths from $u$ to $x$ use it.
\end{proof}

We obtain some immediate consequences of Lemma~\ref{lemma:M(x)}.


\begin{lemma}\label{lemma:M(x)-acyclic}
For a vertex $x$ of a graph $G$, the set of edges $\M(x)$ induces a forest.
\end{lemma}
\begin{proof}
Let $G_x$ be the subgraph of $G$ induced by the edges in $\M(x)$. By Lemma~\ref{lemma:M(x)}, an edge $e$ belongs to $\M(x)$ if and only if $e=uv$, $u\in L_i(x)$ and $v$ is the only neighbour of $u$ in $L_{i-1}(x)$. In this case, we let $v$ be the \emph{parent} of $u$. Each vertex in $G_x$ has at most one parent in $G_x$, and each edge of $G_x$ is an edge between a vertex and its parent. Thus, $G_x$ is a forest.
\end{proof}

\begin{lemma}\label{lemma:incident}
Let $G$ be a graph and $x$ a vertex of $G$. Then, for any edge $e$ incident with $x$, we have $e\in \M(x)$.
\end{lemma}
\begin{proof}
For every vertex $y$ of $L_1(x)$ (that is, every neighbour of $x$), $x$ is the unique neighbour of $y$ in $L_0(x)=\{x\}$. Thus, the claim follows from Lemma~\ref{lemma:M(x)}.
\end{proof}

We next give a characterization for the cases where no other edge (other than those incident with $x$) belongs to $\M(x)$.

\begin{lemma}\label{lemma:M(x)=incident}
  Let $G$ be a connected graph with a vertex $x$ of $G$. The following two conditions are equivalent.
  \begin{itemize}
  \item[(i)] $\M(x)$ is the set of edges incident with $x$.
  \item[(ii)] For every vertex $y$ of $G$ with $y \in V(G) \setminus (\{x\} \cup L_1(x))$, there exist two shortest paths from $x$ to $y$ sharing at most one edge (the one incident with $x$).
  \end{itemize}
\end{lemma}
\begin{proof}
  If (ii) holds, then for any $i>1$, every vertex in level $L_i(x)$ has at least two neighbours in level $L_{i-1}(x)$, so (i) is implied by Lemma~\ref{lemma:M(x)} and Lemma~\ref{lemma:incident}.

  Conversely, if (i) holds, we can prove (ii) by induction on $d_G(x,y)$. The claim is true for all vertices in $L_2(x)$. Assume it is true for vertices at levels $L_2(x),\ldots,L_i(x)$ and let $y$ be a vertex in $L_{i+1}(x)$. By Lemma~\ref{lemma:M(x)}, $y$ has at least two distinct neighbours, $u$ and $v$, in $L_i(x)$. Consider two shortest paths $P_u$ from $u$ to $x$ and $P_v$ from $v$ to $x$. If (ii) is satisfied, we are done. Otherwise, let $st$ be the edge common to both $P_u$ and $P_v$ that is closest to $u$ and $v$ (with $s\in L_j$ and $t\in L_{j-1}$ for some $j>1$). By the induction hypothesis, there are two shortest paths $P_1$ and $P_2$ from $s$ to $x$ that share at most an edge incident with $x$. Using these two shortest paths, we can easily combine them with $P_u$ and $P_v$ to obtain the two shortest paths from $y$ to $x$, as required.
\end{proof}




  

\section{Basic graph families and bounds}\label{sec:bounds}

In this section, we study $\EM$ for standard graph classes, and its relation with other standard graph parameters.

\subsection{Trees and grids}

\begin{theorem}\label{thm:trees}
Let $G$ be a connected graph with at least one edge. We have $\EM(G)=1$ if and only if $G$ is a tree.
\end{theorem}
\begin{proof}
  In a tree, every edge is a bridge. Thus, by Lemma~\ref{lemma:bridges}, if $G$ is a tree, any vertex $x$ of $G$ is a distance-edge-monitoring set and we have $\EM(G)\leq 1$ (and of course as long as there is an edge in $G$, $\EM(G)\geq 1$).

  For the converse, suppose that $\EM(G)=1$. Then, clearly, $G$ must have at least one edge. Moreover, since all edges of $G$ must belong to $\M(x)$, by Lemma~\ref{lemma:M(x)-acyclic}, $G$ must be a forest. Since $G$ is connected, $G$ is a tree.
\end{proof}

Let $G_{a,b}$ denote the grid of dimension $a\times b$. By Theorem~\ref{thm:trees}, we have $\EM(G_{1,a})=1$. We can compute all other values.

\begin{theorem}\label{thm:grids}
For any integers $a,b\geq 2$, we have $\EM(G_{a,b})=\max\{a,b\}$.
\end{theorem}
  \begin{proof}
  Without loss of generality we assume that $\max\{a,b\}=a$, and $G_{a,b}$ has $a$ rows and $b$ columns.

Clearly, for any vertex $x$, by Lemma~\ref{lemma:M(x)}, $\M(x)$ consists of the edges that are in the same row and column as $x$. This shows that any distance-edge-monitoring set $M$ must contain a vertex of each row and column, and $\EM(G_{a,b})\geq a$.

 To see that $\EM(G_{a,b})\leq a$, we can choose any set $M$ of $a$ vertices containing exactly one vertex of every row and at least one vertex of every column of $G_{a,b}$: $M$ is a distance-edge-monitoring set of $G_{a,b}$ of size $a$.
\end{proof}

\subsection{Connection to arboricity and clique number}

The \emph{arboricity} $\arb(G)$ of a graph $G$ is the smallest number of sets into which $E(G)$ can be partitioned and such that each set induces a forest. The \emph{clique number} $\omega(G)$ of $G$ is the size of a largest clique in $G$.

\begin{theorem}\label{thm:arboricity}
For any graph $G$ of order $n$ and size $m$, we have $\EM(G)\geq\arb(G)$, and thus $\EM(G)\geq\frac{m}{n-1}$ and $\EM(G)\geq\frac{\omega(G)}{2}$.
\end{theorem}
\begin{proof}
By Lemma~\ref{lemma:M(x)-acyclic}, for each vertex $x$ of a distance-edge-monitoring set $M$, $\M(x)$ induces a forest. Thus, to each edge $e$ of $G$ we can assign one of the forests $\M(x)$ such that $e\in \M(x)$. This is a partition of $G$ into $|M|$ forests, and thus $\arb(G)\leq\EM(G)$.

Moreover, it is not difficult to see that $\arb(G)\geq\frac{m}{n-1}$ (since a forest has at most $n-1$ edges) and $\arb(G)\geq\frac{\omega(G)}{2}$ (since a clique of size $k=\omega(G)$ has $k(k-1)/2$ edges but a forest that is a subgraph of $G$ can contain at most $k-1$ of these edges).
\end{proof}

\subsection{Connection to vertex covers and consequences for hypercubes and complete bipartite graphs}

We next see that distance-edge-monitoring sets are relaxations of vertex covers. A set $C$ of vertices is a \emph{vertex cover} of $G$ if every edge of $G$ has one of its endpoints in $C$. The smallest size of a vertex cover of $G$ is denoted by $\VC(G)$.

\begin{theorem}\label{thm:VC}
In any graph $G$ of order $n$, any vertex cover of $G$ is a distance-edge-monitoring set, and thus $\EM(G)\leq\VC(G)\leq n-1$. Moreover, we have $\EM(G)=n-1$ if and only if $G$ is the complete graph of order $n$.
\end{theorem}
\begin{proof}
  Let $C$ be a vertex cover of $G$. By Lemma~\ref{lemma:incident}, for every edge $e$, there is a vertex $x \in C$ with $e\in \M(x)$, thus $C$ is distance-edge-monitoring.

  Moreover, any graph $G$ of order $n$ has a vertex cover of size $n-1$: for any vertex $x$, the set $V(G)\setminus\{x\}$ is a vertex cover of $G$.

  Finally, suppose thet $\EM(G)=n-1$: then also $\VC(G)=n-1$. If $G$ is not connected, we have $\VC(G)\leq n-2$ (starting with $V(G)$ and removing any vertex from each connected component of $G$ yields a vertex cover), thus $G$ is connected. Suppose by contradiction that $G$ is not a complete graph. 
Then, we have two vertices $x, y$ in $G$ such that $xy$ is not an edge of $G$. Then, $V(G)\setminus\{x,y\}$ is a vertex cover of $G$, a contradiction.

  This completes the proof.
\end{proof}

In some graphs, any distance-edge-monitoring set is a vertex cover.

\begin{observation}\label{obs:VC}
If, for every vertex $x$ of a graph $G$, $\M(x)$ consists exactly of the edges incident with $x$, then a set $M$ is a distance-edge-monitoring set of $G$ if and only if it is a vertex cover of $G$.
\end{observation}

Note that Observation~\ref{obs:VC} does not provide a characterization of graphs with $\EM(G)=\VC(G)$. For example, as seen in Theorem~\ref{thm:grids}, for the grid $G_{a,2}$, we have $\EM(G_{a,2})=\VC(G_{a,2})=a$, but for any vertex $x$ of $G_{a,2}$, $\M(x)$ consists of the whole row and column of $G_{a,2}$.

Let $K_{a,b}$ be the complete bipartite graph with parts of sizes $a$ and $b$, and let $H_d$ denote the hypercube of dimension $d$. 
  Both $K_{a,b}$ and $H_d$ satisfy Condition (ii) of Lemma~\ref{lemma:M(x)=incident}. In a connected bipartite graph, the smallest vertex cover consists of the smallest of the two parts. Thus, by Observation~\ref{obs:VC}, we obtain the following.

\begin{corollary}\label{cor:hyper-completebip}
We have $\EM(K_{a,b})=\VC(K_{a,b})=\min\{a,b\}$, and $\EM(H_d)=\VC(H_d)=2^{d-1}$.
\end{corollary}

\subsection{Connection to feedback edge set number}

A \emph{feedback edge set} of a graph $G$ is a set of edges such that removing them from $G$ leaves a forest. The smallest size of a feedback edge set of $G$ is denoted by $\FES(G)$ (it is sometimes called the \emph{cyclomatic number} of $G$).

The following folklore lemma uses the terminology defined in Section~\ref{sec:prelim} (a proof can be found for example in Section~5.3.1 of~\cite{ELW12j} or Section~4.1 of~\cite{KK20}).

\begin{lemma}[\cite{ELW12j,KK20}]\label{lemma:basegraph}
Let $G$ be a graph with $\FES(G)=k\geq 2$. The base graph of $G$ has at most $2k-2$ core vertices, that are joined by at most $3k-3$ edge-disjoint core paths.
\end{lemma}

In other words, Lemma~\ref{lemma:basegraph} says that the base graph of a graph $G$ with $\FES(G)=k\geq 2$ can be obtained from a multigraph $H$ of order at most $2k-2$ and size $3k-3$ by subdividing its edges an arbitrary number of times.

Theorem~\ref{thm:trees} states that for any tree $G$ (that is, a graph with feedback edge set number $0$), we have $\EM(G)\leq \FES(G)+1$. We now show the same bound for graphs $G$ with $\FES(G)\leq 2$.

\begin{theorem}\label{prop:bicyclic}
If $\FES(G)\leq 2$, then $\EM(G)\leq\FES(G)+1$. Moreover, if $\FES(G)\leq 1$, then equality holds.
\end{theorem}
\begin{proof}
  Let $k=\FES(G)$. If $k=0$, $G$ is a tree and we are done by Theorem~\ref{thm:trees}.

  Assume next that $k=1$, that is, $G$ is a connected unicyclic graph. Since $G$ is not a tree, we have $\EM(G)\geq 2$. 
  
  Let $C$ be the unique cycle of $G$ and let $x$ be a vertex of $C$. If $C$ has even length, $\M(x)\cap E(C)$ consists of all edges of $C$ except the two that are incident to the antipodal vertex of $x$ in $C$. Thus, for any two non-adjacent vertices $x$ and $y$ of $C$, $(\M(x)\cup \M(y))\cap E(C)=E(C)$. By Lemma~\ref{lemma:bridges}, all other edges of $G$ are monitored and thus $\{x,y\}$ is distance-edge-monitoring. 
  
If $C$ has odd length, $\M(x)\cap E(C)$ consists of all edges of $C$ except the one that connects the two vertices that are farthest from $x$. Thus, for any two vertices $x$ and $y$ of $C$, $(\M(x)\cup \M(y))\cap E(C)=E(C)$ and again $\{x,y\}$ is distance-edge-monitoring.

Assume finally that $k=2$ and let $G_b$ be the base graph of $G$. By Lemma~\ref{lemma:basegraph}, $G_b$ contains at most two core vertices joined by at most three core paths. By Lemma~\ref{lemma:bridges}, any non-empty distance-edge-monitoring set of $G_b$ is also one of $G$, since all edges of $G$ not present in $G_b$ are bridges.

If one of the core paths of $G_b$ is a core cycle, then in fact $G_b$ must have two core cycles, with either one or two core vertices. Then, select one vertex from each core cycle that is farthest from the core vertex on this cycle. One can check that these two vertices form a distance-edge-monitoring set. 

Otherwise, $G_b$ consists of two core vertices $x,y$ joined by three parallel core paths. If $x,y$ are not adjacent, then $\{x,y\}$ forms a distance-edge-monitoring set of $G_b$. (Note that $x$ monitors the edges of the three core paths that are closer to $x$, and $y$ monitors the edges closer to $y$.) Otherwise, the edge $xy$ is one of the three core paths. Then, one may select $x$ and a middle vertex of each of the two core paths of length at least~$2$. This forms a distance-edge-monitoring set and completes the proof.
\end{proof}

We will now give a weaker (but similar) bound for any value of $\FES(G)$. This improves on our previous bound from~\cite{CALDAM}.

\begin{theorem}\label{thm:fes}
Let $G$ be a graph with $\FES(G)=k$. If $k\geq 3$, then $\EM(G)\leq 2k-2$.
\end{theorem}
\begin{proof}
  Let $G_b$ be the base graph of $G$. By Lemma~\ref{lemma:basegraph}, $G_b$ contains at most $2k-2$ core vertices, that are joined by at most $3k-3$ core paths. 
  By Lemma~\ref{lemma:bridges}, any non-empty distance-edge-monitoring set of $G_b$ is also one of $G$, since all edges of $G$ not present in $G_b$ are bridges. Thus, it is sufficient to construct a distance-edge-monitoring set $M$ of $G_b$ of size at most $2k-2$.

  A first candidate would be to select the set $C$ of all core vertices of $G_b$. However, this might not be a distance-edge-monitoring set, indeed for any odd-length core cycle of $G_b$ its middle edge is not monitored, and for any even-length core cycle of $G_b$ the two edges at the opposite of the core vertex are not monitored. Similarly, if $G_b$ contains two distinct core vertices $c_1,c_2$ that are joined by a core path $P$ of odd length at least~$3$ and by another core path of length~$1$ (an edge), then the middle edge of $P$ is not monitored. (Note that every other edge $e$ is monitored, indeed there is one endpoint $x$ of $e$ such that the unique shortest path from $x$ to its closest core vertex goes through $e$.) We call these core paths, \emph{problematic} core paths. Let $\mathcal P$ be the set of problematic core paths. In what follows, we will start from $C$ as a solution set and adjust it to handle problematic core paths.

  To this end, we build a subgraph $G'_b$ of $G_b$, by iteratively selecting a new core path $P$ of $\mathcal P$, and by removing the inner-vertices and edges of $P$ from $G_b$, unless we decrease the number of core vertices of the graph by more than~$1$. When there is no such candidate core path left in $\mathcal P$, we stop the process and obtain the graph $G'_b$. Let $p$ be the number of core paths of $\mathcal P$ that we have removed from $G_b$ in this process, and let $k'=\FES(G'_b)$.

  First, we show that $k'=k-p$. First, to see that $k'\geq k-p$, we show that $k\leq k'+p$. Observe that one can obtain a feedback edge set of $G_b$ from one of $G'_b$ by additionally selecting one edge of each of the $p$ core paths that were deleted from $G_b$ to obtain $G'_b$, which shows that $k\leq k'+p$, as claimed. Second, to see that $k'\leq k-p$, consider any feedback edge set $F$ of $G_b$, and note that it necessarily contains one edge of each core cycle of $G_b$. Moreover, for every pair $c_1,c_2$ of distinct core vertices of $G_b$ joined by at least~$2$ core paths, it contains an edge of all of them, except possibly one. In fact we may assume that $F$ is chosen so that it contains an edge of all of them that have length at least~$2$. Thus, every problematic core path of $G_b$ has an edge in $F$. Removing these $p$ edges from $F$, we obtain the desired feedback edge set of $G'_b$ of size $k-p$.

 Let $C'$ be the set of core vertices of $G'_b$. Since at each step of the procedure used for building $G'_b$, we decrease the number of core vertices by at most~$1$, $G'_b$ has at least $|C|-p$ core vertices. Thus, since $k'=k-p$ and applying Lemma~\ref{lemma:basegraph} to $G'_b$, we have:

  $$
  \begin{array}{rcl}
    |C| & \leq & |C'|+p \\
    & \leq & 2k'-2+p\\
    & = & 2(k-p)-2+p\\
    & = & 2k-p-2
  \end{array}
  $$  

  We now build our distance-edge-monitoring set $M$ as follows. First, we let $M=C$. Moreover, for each problematic core path $P$ that was deleted to obtain $G'_b$, we add one arbitrary inner-vertex of $P$ to $M$. At this point, $M$ monitors all edges, except those of the problematic paths of $G_b$ that were not deleted when constructing $G'_b$. Observe that in $G'_b$, each such core path $P$ joins two distinct adjacent core vertices $c_1,c_2$ of $G'_b$, both having degree exactly~$3$ in $G'_b$. (Moreover, since $k\geq 3$, $c_1$ and $c_2$ are joined only by two core paths in $G'_b$.) We select an arbitrary core vertex among $c_1$ and $c_2$ (say $c_1$), remove it from $M$, and replace it with the neighbour of $c_1$ on $P$. 
  We do this for every such remaining problematic core path. Again, since $k\geq 3$, now every problematic core path conatains a vertex of $M$. 

It is clear that $M$ has size at most $2k-2$. We now show that $M$ is distance-edge-monitoring. To do so, consider an edge $e$ of $G_b$, which lies on some core path $P$. If $P$ is not a problematic path or $e$ is not the middle edge of $P$, then similarly as for the set $C$, there is an endpoint $x$ of $e$ whose unique shortest path to one of the core vertices $c$ of $P$ (or the neighbour of $c$, if $c$ was removed from $M$) goes through $e$, and so $e$ is monitored by $c$ or its neighbour. For the remainder, assume that $P$ is problematic and that $e$ is the middle edge of $P$. If $P$ was among the $p$ paths from which we added an extra vertex $x$ to $M$, then $x$ monitors $e$. Otherwise, we have removed one of the core vertices of $P$, and replaced it with its neighbour on $P$. Then, this neighbour monitors $e$. This completes the proof.
\end{proof}

 We do not believe that the bound of Theorem~\ref{thm:fes} is tight. There are examples of graphs $G$ where $\EM(G)=\FES(G)+1$, for example this is the case for the grid $G_{a,2}$: by Theorem~\ref{thm:grids}, when $a\geq 2$ we have $\EM(G_{a,2})=a$, and $\FES(G_{a,2})=a-1$.

 Note that $G_{a,b}$ for $a,b\geq 2$ provides examples of a family of graphs where for increasing $a,b$ the difference between $\FES(G_{a,b})$ and $\EM(G_{a,b})$ is unbounded by a constant. Indeed, by Theorem~\ref{thm:grids}, we have $\EM(G_{a,b})=\max\{a,b\}$, but $\FES(G_{a,b})$ is linear in the order $ab$.

\section{Graphs $G$ with $\EM(G)=2$}\label{sec:dem=2}

In this section, we characterize connected graphs $G$ with $\EM(G)=2$.

For two vertices $u,v$ of a graph $G$ and two non-negative integers $i,j$, we denote by $B_{i,j}(u,v)$ the set of vertices at distance $i$ from $u$ and distance $j$ from $v$ in $G$.

\begin{observation}
Let $G$ be a graph, and let $u,v$ be two vertices in $G$. Let $i,j$ be two non-negative integers such that $B_{i,j}(u,v) \neq \emptyset$. Then, $|i-j|\leq d_G(u,v)$. 
\end{observation}

\begin{proof}
Let $x \in B_{i,j}(u,v)$. According to the definition of $B_{i,j}(u,v)$, $d_G(u,x)$ = $i$ and $d_G(v,x)$ = $j$. According to the triangle inequality, it holds that $i \leq d_G(u,v) + j$. Hence, $i - j \leq d_G(u,v)$. Similarly, it holds that $j \leq d_G(u,v) + i$. Hence, $j - i \leq d_G(u,v)$. Overall, $|i-j|\leq d_G(u,v)$.
\end{proof}

\begin{observation}\label{obs:i'j'}
Let $G$ be a graph, and let $u,v$ be two vertices in $G$. Let $i,j$ be two non-negative integers such that $B_{i,j}(u,v) \neq \emptyset$. Let $x \in B_{i,j}(u,v)$, and let $y \in B_{i',j'}(u,v)$ be a neighbor of $x$. Then, $i'\in \{i-1, i, i+1\}$ and $j'\in \{j-1, j, j+1\}$. 
\end{observation}

\begin{proof}
Let $x \in B_{i,j}(u,v)$ and let $y \in B_{i',j'}(u,v)$ be a neighbor of $x$. Because $y$ is a neighbor of $x$, the distances from $u$ to $y$ and $x$ can only differ by at most 1. Hence, $i'\in \{i-1, i, i+1\}$. Likewise, because $y$ is a neighbor of $x$, the distances from $v$ to $y$ and $x$ can only differ by at most 1. Hence, $j'\in \{j-1, j, j+1\}$.
\end{proof}
We are now ready to state our characterization in the next theorem. An example of a graph $G$ with $\EM(G)=2$ following the characterization is given in Figure~\ref{fig:example2}.

\begin{theorem}\label{thm:dem2}
Let $G$ be a connected graph with at least one cycle, and let $G_b$ be the base graph of $G$. Then, $\EM(G) = 2$ if and only if there are two vertices $u,v$ in $G_b$ such that all of the following conditions (I) -- (IV) hold in $G_b$.
 	\begin{itemize}

 	\item[(I)] $\forall i,j \in \{0,1,2,\dots\}$: $B_{i,j}(u,v)$ is an independent set.

	\item[(II)] $\forall i,j \in \{1,2,3,\dots\}$: Every vertex $x$ in $B_{i,j}(u,v)$ has at most one neighbor in each of the four sets $B_{i-1,j}(u,v)\cup B_{i-1,j-1}(u,v)$, $B_{i-1,j}(u,v)\cup B_{i-1,j+1}(u,v)$, $B_{i,j-1}(u,v)\cup B_{i-1,j-1}(u,v)$ and $B_{i,j-1}(u,v)\cup B_{i+1,j-1}(u,v)$. 

	\item[(III)] $\forall i,j \in \{1,2,3,\dots\}$: There is no 4-vertex path $zxyz'$ with $z\in B_{i-1,a}(u,v)$, $z'\in B_{a',j}(u,v)$, $x\in B_{i,j}(u,v)$, $y\in B_{i-1,j+1}(u,v)$, $a\in \{j-1,j+1\}$, $a'\in \{i-2,i\}$.
	
	\item[(IV)] $\forall i,j \in \{1,2,3,\dots\}$: $x\in B_{i,j}(u,v)$ has neighbors in at most two sets among $B_{i-1,j+1}(u,v)$, $B_{i-1,j-1}(u,v)$, $B_{i+1,j-1}(u,v)$. 
		 	
	\end{itemize}

\end{theorem}

\begin{proof}
$\Rightarrow$
Let us assume $\EM(G) = 2$, hence $\EM(G_b)=2$ according to Observation~\ref{obs:EM2}. Let $\{u,v\}$ be a distance-edge-monitoring set of $G_b$. We will show that the properties (I) - (IV) hold. 

 	\begin{itemize}

 	\item[(I)] Assume that $B_{i,j}(u,v)$ is not an independent set. Let $e=xy$  be an edge such that $x,y \in B_{i,j}(u,v)$. Then, the distance from $x$ to $u$ is $i$ and the distance from $y$ to $u$ is $i$, hence $e$ is not monitored by $u$, according to Lemma~\ref{lemma:M(x)}. Likewise, $e$ is not monitored by $v$, a contradiction. 

	\item[(II)] Let $x \in B_{i,j}(u,v)$.
		\begin{itemize}
			\item[(i)]Assume $x$ has two neighbors $y,y' \in B_{i-1,j}(u,v) \cup B_{i-1,j-1}(u,v)$. Assume first that $x$ has two neighbors $y,y'$ such that either both $y,y' \in B_{i-1,j}(u,v)$, both $y,y'\in B_{i-1,j-1}(u,v)$, or $y \in B_{i-1,j}(u,v)$, $y' \in B_{i-1,j-1}(u,v)$. Let $e=xy$. Then, the distance from $y$ to $u$ and the distance from $y'$ to $u$ are both $i-1$, hence $e$ is not monitored by $u$, according to Lemma~\ref{lemma:M(x)}. The distance from $y$ to $v$ and the distance from  $x$ to $v$ are also the same, hence $e$ is not monitored by $v$, according to Lemma~\ref{lemma:M(x)}. This is a contradiction.
			
			
			\item[(ii)]Assume $x$ has two neighbors $y,y' \in B_{i-1,j}(u,v) \cup B_{i-1,j+1}(u,v)$. If both $y,y'\in B_{i-1,j}(u,v)$ we are done by Case~(i). Thus, assume that $y\in B_{i-1,j}(u,v) \cup B_{i-1,j+1}(u,v)$ and $y'\in B_{i-1,j+1}(u,v)$, and let $e=xy$. If $y,y' \in B_{i-1,j+1}(u,v)$, as in Case~(i), the ordered pair of distances from $y$ to $u$ and $v$ and the ordered pair of distances from $y'$ to $u$ and $v$ are the same, and by Lemma~\ref{lemma:M(x)} $e$ is not monitored. If $y \in B_{i-1,j}(u,v)$, $y' \in B_{i-1,j+1}(u,v)$, the same argument works to show that $u$ does not monitor $e$. Moreover, the distances from $x$ to $v$ and from $y$ to $v$ are both $j$, and by Lemma~\ref{lemma:M(x)} $v$ also does not monitor $e$. This is a contradiction.
			
			\item[(iii)]The cases where $x$ has two neighbors in $B_{i,j-1}(u,v) \cup B_{i-1,j-1}(u,v)$ or in $B_{i,j-1}(u,v) \cup B_{i+1,j-1}(u,v)$ are completely symmetric to Cases~(i) and~(ii).

		\end{itemize}
	\item[(III)] Assume there is a path $zxyz'$ with $z\in B_{i-1,a}(u,v)$, $z'\in B_{a',j}(u,v)$, $x\in B_{i,j}(u,v)$, $y\in B_{i-1,j+1}(u,v)$, $a\in \{j-1,j+1\}$, $a'\in \{i-2,i\}$, and let $e=xy$, then the distance from $u$ to $z$ is $i-1$, and the distance from $u$ to $y$ is $i-1$, hence by Lemma~\ref{lemma:M(x)} $e$ is not monitored by $u$. The distance from $v$ to $z'$ is $j$, and the distance from $v$ to $x$ is $j$, hence $e$ is also not monitored by $v$ according to Lemma~\ref{lemma:M(x)}, a contradiction.
	\item[(IV)]
		
			Let us assume by contradiction that $x \in B_{i,j}(u,v)$ has three neighbors $y,y',z$ such that $y \in B_{i-1,j-1}(u,v)$, $y' \in B_{i+1,j-1}(u,v)$, and $z \in B_{i-1,j+1}(u,v)$, and let $e=xy$. The distance from $z$ to $u$ is $i-1$, and the distance from $y$ to $u$ is $i-1$, hence $e$ is not monitored by $u$. The distance from $y$ to $v$ is $j-1$, and the distance from $y'$ to $v$ is $j-1$, hence $e$ is not monitored by $v$, according to Lemma~\ref{lemma:M(x)}, a contradiction.  
	\end{itemize}

\noindent $\Leftarrow$  Let us assume that there are two vertices $u,v$ in $G_b$ such that all of the conditions (I) - (IV) hold in $G_b$. We will show that $\{u,v\}$ is a distance-edge-monitoring set in $G_b$, and hence $\EM(G)=2$ according to Observation~\ref{obs:EM2}. Let $e=xy$ be an edge in $G$ such that $x \in B_{i,j}(u,v)$. According to condition (I), $y \notin B_{i,j}(u,v)$. We will distinguish three cases with respect to $y$. All the other cases are symmetric to these three main cases.  
	
	\begin{itemize}
	\item[Case(a):] $y \in B_{i-1,j-1}(u,v)$. (The case $y \in B_{i+1,j+1}(u,v)$ is symmetric.)

	By contradiction, assume $e$ is not monitored. Then, there is a path $P_i$ of length $i$ from $x$ to $u$ not using $e$, and there is a path $P_j$ of length $j$ from $x$ to $v$ not using $e$. Let $z$ be the neighbor of $x$ in $P_i$. Then, $z$ belongs to $B_{i-1,a}(u,v)$. Let $z'$ be the neighbor of $x$ in $P_j$. Then, $z'$ belongs to $B_{a',j-1}(u,v)$. According to Observation~\ref{obs:i'j'} and condition (II), $a=j+1$ and $a'=i+1$. Hence, $x$ has neighbors in all three sets $B_{i-1,j+1}(u,v)$, $B_{i-1,j-1}(u,v)$, $B_{i+1,j-1}(u,v)$, in contradiction to condition (IV). 

	\item[Case(b):] $y \in B_{i,j-1}(u,v)$. (The cases $y \in B_{i-1,j}(u,v)$, $B_{i+1,j}(u,v)$ or $B_{i,j+1}(u,v)$ are symmetric.)
	
	By contradiction, assume $e$ is not monitored. Then, there is a path $P_j$ of length $j$ from $x$ to $v$ not using $e$. Let $z'$ be the neighbor of $x$ in $P_j$. Then, $z'$ belongs to $B_{a',j-1}(u,v)$. As $a' \in \{i-1, i, i+1\}$, this contradicts condition (II).

	\item[Case(c):] $y \in B_{i-1,j+1}(u,v)$. (The case $y \in B_{i+1,j-1}(u,v)$ is symmetric.)
	
	By  contradiction, assume $e$ is not monitored. Then, there is a path of length $i$ from $x$ to $u$ avoiding $y$ (let $z$ be the neighbor of $x$ on this path). Similarly, there is a path of length $j+1$ from $y$ to $v$ avoiding $x$ (let $z'$ be the neighbor of $y$ on this path). Thus, there is a 4-vertex path $zxyz'$ with $z \in B_{i-1, a}(u,v)$, $z' \in B_{a',j}(u,v)$, $a \in \{j-1,j,j+1\}$, $a' \in \{i-2,i-1,i\}$. According to condition (II), $a \neq j$ and $a' \neq i-1$. However, the other cases ($a \in \{j-1,j+1\}$ and $a' \in \{i-2,i\}$) contradict condition (III).
	\end{itemize}
\end{proof}

\begin{figure}[!htpb]
\centering
\scalebox{1}{\begin{tikzpicture}[every loop/.style={},scale=1]
  \node[blackvertex,label={180:$u$}] (u) at (0,0) {};
  \node[whitevertex] (b13) at (1,0) {};
  \node[whitevertex] (b13') at (1,0.5) {};
  \node[whitevertex] (b22) at (2,0) {};
  \node[whitevertex] (b31) at (3,0) {};
  \node[whitevertex] (b31') at (3,0.5) {};
  \node[blackvertex,label={0:$v$}] (v) at (4,0) {};
  \node[whitevertex] (b14) at (0,1) {};
  \node[whitevertex] (b24) at (0,2) {};
  \node[whitevertex] (b33) at (2,2) {};
  \node[whitevertex] (b33') at (2,2.5) {};
  \node[whitevertex] (b32) at (2.75,2.75) {};
  \node[whitevertex] (b41) at (3.5,2) {};
  \node[whitevertex] (b41') at (4,2) {};  

  \draw[thick] (b14) -- (b13) -- (u) -- (b14) -- (b24) -- (b33') -- (b32) -- (b22) -- (b33) -- (b24) (b13) -- (b22) -- (b31') -- (b32) (u) -- (b13') -- (b22) -- (b31) -- (v) -- (b31') -- (b41) -- (v) -- (b41') -- (b31');

\draw[line width=0.8pt] (b14) ellipse (0.5cm and 0.3cm);
\draw[line width=0.8pt] (b24) ellipse (0.5cm and 0.3cm);  
\draw[line width=0.8pt] (b22) ellipse (0.3cm and 0.5cm);
\draw[line width=0.8pt] (1,0.25) ellipse (0.3cm and 0.7cm);
\draw[line width=0.8pt] (2,2.25) ellipse (0.3cm and 0.7cm);
\draw[line width=0.8pt] (b32) ellipse (0.3cm and 0.5cm);
\draw[line width=0.8pt] (3.75,2) ellipse (0.7cm and 0.3cm);
\draw[line width=0.8pt] (3,0.25) ellipse (0.3cm and 0.7cm);

\node[xshift=-9mm] at (b14) {$B_{1,4}$};
\node[xshift=-9mm] at (b24) {$B_{2,4}$};
\node[yshift=12mm] at (b33) {$B_{3,3}$};
\node[yshift=8mm] at (b32) {$B_{3,2}$};
\node[xshift=14mm] at (b41) {$B_{4,1}$};
\node[yshift=-8mm] at (b13) {$B_{1,3}$};
\node[yshift=-8mm] at (b22) {$B_{2,2}$};
\node[yshift=-8mm] at (b31) {$B_{3,1}$};

\end{tikzpicture}}
\caption{A (base) graph $G$ with $\EM(G)=2$, where $\{u,v\}$ is a distance-edge-monitoring set. For easier readability, we use the notation $B_{i,j}$ for $B_{i,j}(u,v)$.}\label{fig:example2}
\end{figure}
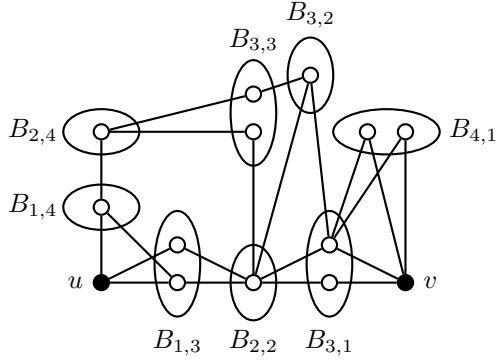

\section{Complexity}\label{sec:complex}

We now study the algorithmic complexity of finding an optimal distance-edge-monitoring set.

A \emph{$c$-approximation algorithm} for a given optimization problem is an algorithm that returns a solution whose size is always at most $c$ times the optimum. We refer to the books~\cite{ACGKMP99,Hro04} for more details.
For a decision problem $\Pi$ and for some parameter $p$ of the instance, an algorithm for $\Pi$ is said to be \emph{fixed parameter tractable} (fpt for short) if it runs in time $f(p)n^c$, where $f$ is a computable function, $n$ is the input size, and $c$ is a constant. In this paper, we will always consider the solution size $k$ as the parameter. The class FPT contains the parameterized decisions problems solvable by an fpt algorithm. The classes W[$i$] (with $i\geq 1$) denote complexity classes with parameterized decision problems that are believed not to be fpt. We refer to the books~\cite{DF13,N06} for more details.

  \subsection{NP-hardness}

We will now use the connection to vertex covers hinted in Section~\ref{sec:bounds} to derive an NP-hardness result. For two graphs $G$, $H$, $G\bowtie H$ denotes the graph obtained from disjoint copies of $G$ and $H$ with all possible edges between $V(G)$ and $V(H)$. We denote by $K_1$ the graph on one single vertex.

\begin{theorem}\label{thm:VC-universal-vertex}
For any graph $G$, we have $\VC(G)\leq\EM(G\bowtie K_1)\leq\VC(G)+1$. Moreover, if $G$ has radius at least $4$, then $\VC(G)=\EM(G\bowtie K_1)$.
\end{theorem}
\begin{proof}
We denote by $u$ the vertex from $K_1$. By Theorem~\ref{thm:VC}, we have $\EM(G\bowtie K_1)\leq\VC(G\bowtie K_1)\leq\VC(G)+1$ since $u$ together with any vertex cover of $G$ is a vertex cover of $G\bowtie K_1$.

To see that $\VC(G)\leq\EM(G\bowtie K_1)$, assume that we have a distance-edge-monitoring set $M$ of $G\bowtie K_1$, and let $vw$ be any edge of $G$. There exists a vertex $x$ in $M$ such that $vw\in \M(x)$. Then, we must have $x\in\{v,w\}$, indeed, if not, then $vw$ would be part of a unique shortest path of length~$2$ from $x$ to some vertex $y$, which is not possible since the path from $x$ to $y$ going through $u$ is another path of length~$2$.

Now, assume that $G$ has radius at least~$4$: it remains to prove that $\EM(G\bowtie K_1)\leq\VC(G)$. Let $M$ be a vertex cover of $G$. By Lemma~\ref{lemma:incident}, all edges incident with a vertex of $M$ are monitored by $M$, which includes all (original) edges of $G$. Let $e=uv$ be an edge incident with the vertex $u$ of $K_1$ such that $v\notin M$. We claim that there is a vertex $x$ of $M$ with $d_G(v,x)\geq 3$. Indeed, since $G$ has radius at least $4$, there is a vertex at distance~$4$ of $v$ in $G$. Consider a neighbour of that vertex that is itself at distance~$3$ of $v$ in $G$. Since $M$ is a vertex cover, one of these two vertices is in $M$, which proves the existence of $x$. Now, we know that $d_{G\bowtie K_1}(v,x)=2$, and the only path of length~$2$ from $v$ to $x$ goes through $e$. Thus $e\in \M(x)$. This shows that all edges are monitored by $M$, and completes the proof.
\end{proof}


\begin{corollary}\label{cor:NP-h}
\EMdec is NP-complete, even for graphs obtained from a planar subcubic graph by attaching a universal vertex.
\end{corollary}
\begin{proof}
It is known (see~\cite{GJ77}) that \textsc{Vertex Cover} is NP-complete for planar subcubic graphs with radius at least~$4$ (for the latter condition, note that a subcubic graph of radius at most~$3$ has constant order). Theorem~\ref{thm:VC-universal-vertex} provides a polynomial reduction from such instances, and this completes the proof.
\end{proof}

\subsection{Approximation algorithm}

Given a hypergraph $H=(X,S)$ with vertex set $X$ and edge set $S$, a \emph{set cover} of $H$ is a subset $C\subseteq S$ of edges such that each vertex of $X$ belongs to at least one edge of $C$. We will provide reductions for \EMopt to and from \SCdec and \SCopt.

\decisionpb{\SCdec}{A hypergraph $H=(X,S)$, an integer $k$.}{Is there a set cover $C\subseteq S$ of size at most $k$?}{0.9}

\optimpb{\SCopt}{A hypergraph $H=(X,S)$.}{Build a smallest possible set cover of $H$.}{0.9}

It is known that \SCopt is approximable in polynomial time within a factor of $\ln(|X|+1)$~\cite{J74}, but, unless P$=$NP, not within a factor of $(1-\epsilon)\ln |X|$ (for every positive $\epsilon$)~\cite{DS14}. Moreover, \SCdec is W[2]-hard (when parameterized by $k$) and not solvable in time $|X|^{o(k)}|H|^{O(1)}$ unless FPT$=$W[1]~\cite{n^o(k)}.

\begin{theorem}
\EMopt is approximable within a factor of $\ln(|E(G)|+1)$ in polynomial time.
\end{theorem}
\begin{proof}
  Let $G$ be a graph and consider the following hypergraph $H=(X,S)$ with $X=E(G)$ and such that $S$ contains, for each vertex $x$ of $G$, the set $S_x=\{e\in X~|~e\in \M(x)\}$. Now, it is not difficult to see that there is a one-to-one correspondence between set covers of $H$ and distance-edge-monitoring sets of $G$, where we associate to each vertex $x$ of a distance-edge-monitoring set of $G$, the set $S_x$ in a set cover of $H$. Thus, the result follows from the $\ln(|X|+1)$-approximation algorithm for \SCopt from~\cite{J74}.
\end{proof}

\subsection{Approximation and parameterized hardness}

\begin{theorem}\label{thm:approx-hardness}
  Even for graphs $G$ that are (a) of diameter~$4$, (b) bipartite and of diameter 6, or (c) bipartite and of maximum degree~$3$, \EMopt is not approximable within a factor of $(1-\epsilon)\ln |E(G)|$ in polynomial time, unless $P=NP$. Moreover, for such instances, \EMdec cannot be solved in time $|G|^{o(k)}$, unless FPT$=$W[1], and it is W[2]-hard for parameter $k$.
\end{theorem}
\begin{proof}
For an instance $(H,k)$ of \SCdec, we will construct in polynomial time instances $(G,k+2)$ or $(G,k+1)$ of \EMdec so that $H$ has a set cover of size $k$ if and only if $G$ has a distance-edge-monitoring set of size at most $k+2$ or $k+1$.

In our first reduction, the obtained instance has diameter~$4$, while in our second reduction, the obtained instance is bipartite and has diameter~$6$, and in our third reduction, the graph $G$ is bipartite and has maximum degree~$3$. The three constructions are similar.

The statement will follow from the hardness of approximating \SCopt proved in~\cite{DS14}, the parameterized hardness of \SCdec (parameterized by solution size), and the lower bound on its running time~\cite{n^o(k)}.

First of all, we point out that we may assume that in an instance $(H=(X,S),k)$ of \SCdec, there is no vertex of $X$ that belongs to a unique set of $S$. Indeed, otherwise, we are forced to take $S$ in any set cover of $H$; thus, by removing $S$ and all vertices in $S$, we obtain an equivalent instance $(H',k-1)$. We can iterate until the instance satisfies this property.

\paragraph{(1)}
We now describe the first reduction, in which the obtained instance has diameter~4.
Let $(H,k)=((X,S),k)$ be an instance of \SCdec, where $X= \{x_1,x_2,\ldots,x_{|X|}\}$,
$S = \{C_1,C_2,\ldots,\allowbreak C_{|S|}\}$ and $C_i=\{c_{i,j} \mid x_j \in C_i\}$.
Construct the following instance $(G,k+2)=((V,E),k+2)$ of \EMdec,
where $V=V_1 \cup V_2 \cup \cdots \cup V_5 \cup V'_1 \cup V'_2 \cup V'_3$, 
$E= E_1 \cup E_2 \cup E_3 \cup E_4 \cup E'_1 \cup E'_2 \cup E'_3 \cup E'_4$ and

\allowdisplaybreaks

  \begin{eqnarray*}
V_1 & = & \{u_1,u_2,u_3\}, \\
V_2 & = & \{v_i \mid 1 \leq i \leq |S|\}, \\
V_3 & = & S, \\
V_4 & = & \{ c_{i,j} \mid 1 \leq i \leq |S|, 1 \leq j \leq |X|, x_j \in C_i\}, \\
V_5 & = & X, \\[6pt]
V'_1 & = & \{u'_1,u'_2,u'_3\}, \\
V'_2 & = & \{ c'_{i,j} \mid 1 \leq i \leq |S|, 1 \leq j \leq |X|, x_j \in C_i\}, \\
V'_3 & = & \{w'_j \mid 1 \leq j \leq |X|\}, \\[6pt]
E_1 & = & \{(u_1,u_2), (u_1,u_3), (u_2,u'_1), (u_3,u'_1)\}, \\
E_2 & = & \{(u_1,v_i), (v_i,C_i) \mid 1 \leq i \leq |S|\}, \\
E_3 & = & \{(C_i, c_{i,j}) \mid 1 \leq i \leq |S|, 1 \leq j \leq |X|,x_j \in C_i\}, \\
E_4 & = & \{(c_{i,j},x_j) \mid 1 \leq i \leq |S|, 1 \leq j \leq |X|, x_j \in C_i\}, \\[6pt]
E'_1 & = & \{(u'_1,u'_2), (u'_2,u'_3), (u'_3,u'_1)\}, \\
E'_2 & = & \{(u'_1,c'_{i,j}), (c'_{i,j},c_{i,j}) \mid 1 \leq i \leq |S|, 1 \leq j \leq |X|, x_j \in C_i\}, \\
E'_3 & = & \{(u'_1,w'_j), (w'_j,x_j) \mid 1 \leq j \leq |X|\}, \\
E'_4 & = & \{(u'_1,v_i) \mid 1 \leq i \leq |S|\}.
\end{eqnarray*}

An example is given in Figure~\ref{fig:reduction-diam4}. Notice that the purpose of the vertices in $V'_1,V'_2,V'_3$ is only to reduce the diameter. A simpler reduction without these vertices, and with an additional edge joining $u_2$ and $u_3$ would also function.

\begin{figure}[!htpb]
\centering
\scalebox{0.8}{\begin{tikzpicture}[every loop/.style={},scale=1]
  \node[blackvertex,label={90:$v_1$}] (v1) at (2,2) {};
  \node[blackvertex,label={90:$v_2$}] (v2) at (2,0) {};
  \node[blackvertex,label={90:$v_3$}] (v3) at (2,-2) {};
  \node[blackvertex,label={100:$C_1$}] (C1) at (4,2) {};
  \node[blackvertex,label={100:$C_2$}] (C2) at (4,0) {};
  \node[blackvertex,label={100:$C_3$}] (C3) at (4,-2) {};
  \node[blackvertex,label={90:$c_{1,1}$}] (c11) at (6,3) {};
  \node[blackvertex,label={90:$c_{1,2}$}] (c12) at (6,2) {};
  \node[blackvertex,label={90:$c_{3,1}$}] (c31) at (6,1) {};
  \node[blackvertex,label={90:$c_{2,2}$}] (c22) at (6,0) {};
  \node[blackvertex,label={90:$c_{2,3}$}] (c23) at (6,-1) {};
  \node[blackvertex,label={90:$c_{2,4}$}] (c24) at (6,-2) {};
  \node[blackvertex,label={90:$c_{3,3}$}] (c33) at (6,-3) {};
  \node[blackvertex,label={90:$c_{3,4}$}] (c34) at (6,-4) {};
    \node[blackvertex,label={0:$x_1$}] (x1) at (8,2) {};
  \node[blackvertex,label={0:$x_2$}] (x2) at (8,0) {};
  \node[blackvertex,label={0:$x_3$}] (x3) at (8,-2) {};
  \node[blackvertex,label={0:$x_4$}] (x4) at (8,-4) {};
   \node[blackvertex,label={180:$u_1$}] (u1) at (0,0) {}; 
   \node[blackvertex,label={180:$u_2$}] (u2) at (-1,-4) {}; 
   \node[blackvertex,label={0:$u_3$}] (u3) at (1,-4) {}; 
   \node[blackvertex,label={270:$u'_1$}] (u'1) at (4,-7) {}; 
   \node[blackvertex,label={180:$u'_2$}] (u'2) at (3,-8) {}; 
   \node[blackvertex,label={0:$u'_3$}] (u'3) at (5,-8) {}; 

   \node[blackvertex,label={}] (v'2) at (6,-5.5) {}; 
   \node[blackvertex,label={}] (v'3) at (8,-5.5) {}; 
   \path (7,-5.5) node {\ldots};

  \draw[thick] (u1) -- (v1) -- (C1);
  \draw[thick] (u1) -- (v2) -- (C2);
  \draw[thick] (u1) -- (v3) -- (C3);
  \draw[thick] (u3) -- (u1) -- (u2);
  \draw[thick] (u3) -- (u'1) -- (u2);
  \draw[thick] (u'3) -- (u'1) -- (u'2) -- (u'3);
  \draw[thick] (C1) -- (c11) -- (x1);
  \draw[thick] (C1) -- (c12) -- (x2);
  \draw[thick] (C2) -- (c22) -- (x2);
  \draw[thick] (C2) -- (c23) -- (x3);
  \draw[thick] (C2) -- (c24) -- (x4);
  \draw[thick] (C3) -- (c33) -- (x3);
  \draw[thick] (C3) -- (c34) -- (x4);
  \draw[thick] (C3) -- (c31) -- (x1);
  \draw[thick] (u'1) -- (v1);
  \draw[thick] (u'1) -- (v2);
  \draw[thick] (u'1) -- (v3);
  \draw[thick] (u'1) -- (v'2) -- ++(0,0.5);
  \draw[thick] (u'1) -- (v'3) -- ++(0,0.5);
  
  
\end{tikzpicture}}
\caption{First reduction from \SCdec to \EMdec from the proof of Theorem~\ref{thm:approx-hardness} applied to the hypergraph $(\{x_1,x_2,x_3,x_4\},\{C_1=\{x_1,x_2\},C_2=\{x_2,x_3,x_4\},C_3=\{x_1,x_3,x_4\}\})$. Vertices and edges of $V'_i$ and $E'_i$ for $i=2,3$ are only suggested.}\label{fig:reduction-diam4}
\end{figure}
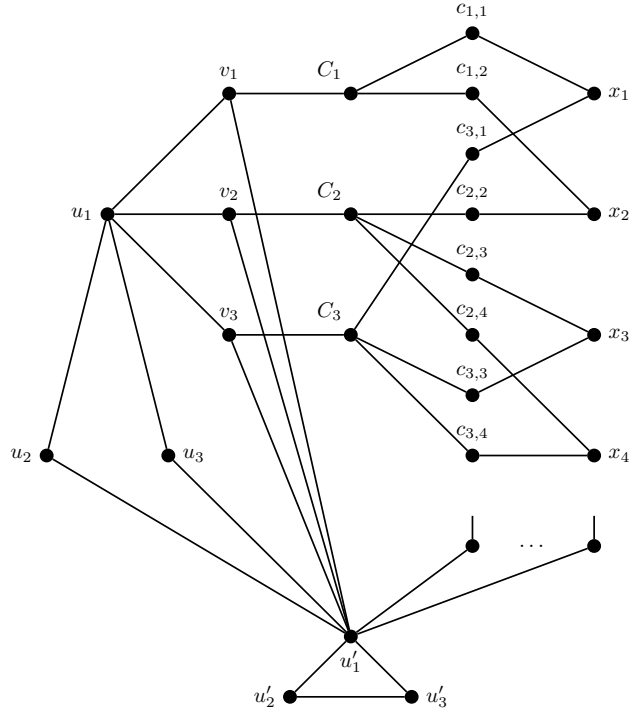

To see that $G$ has diameter~$4$, observe that every vertex of $G$ has a path of length at most~$2$ to the vertex $u_1'$.

Let $C$ be a set cover of $H$ of size $k$. Define $M = C \cup \{u_1,u'_2\}$. Then, by Lemma~\ref{lemma:M(x)}, $u_1$ monitors (in particular) the edges $(u_1,u_2)$, $(u_1,u_3)$, $(u'_1,u'_3)$, and all the edges in $E_2 \cup E_3$. Similarly, $u'_2$ monitors the edges $(u'_1,u'_2)$, $(u'_2,u'_3)$, $(u_3,u'_1)$, $(u_2,u'_1)$ and all the edges in $E'_2 \cup E'_3 \cup E'_4$. It thus remains to show that all edges of $E_4$ are monitored. Notice that among those edges, vertex $C_i$ of $S$ monitors exactly all edges $c_{j_1,j_2}$ with $x_{j_2}\in C_i$. Thus, if $e=(c_{i,j},x_j)$, there is $i'$ such that $x_j \in C_{i'}$ and $C_{i'} \in C$, and $e$ is monitored by $C_{i'}$ (either $i=i'$ and the only shortest path from $x_j$ to $C_i$ contains $e$, or $i\neq i'$ and the only shortest path from $c_{i,j}$ to $C_{i'}$ contains $e$). Hence, $M$ is a distance-edge-monitoring set of $G$ of size at most $k+2$.

Conversely, let $M$ be a distance-edge-monitoring set of $G$ of size at most $k+2$. 
In order to monitor the edge $(u'_2,u'_3)$, either $u'_2\in M$ or $u'_3 \in M$. 
In order to monitor the edges $(u_1,u_3)$ and $(u_1,u_2)$, there must be a vertex of $M$ in $\{u_1,u_2,u_3\}$. We may replace $M\cap\{u_1,u_2,u_3\}$ by $u_1$ and $M\cap\{u'_2,u'_3\}$ by $u'_2$, as $u_1$ monitors the same edges as $\{u_1,u_2,u_3\}$ (among those not already monitored by $u'_2$) and $u'_2$ monitors the same edges as $\{u'_2,u'_3\}$ (among those not already monitored by $u_1$). As seen in the previous paragraph, all edges of $E_1$, $E_2$, $E_3,$ $E'_1$, $E'_2$, $E'_3$ and $E'_4$ are monitored by $\{u_1,u'_2\}$. However, no edge of $E_4$ is monitored by any vertex of $V_1\cup V'_1$. Thus, all remaining vertices of $M$ are needed precisely to monitor the edges of $E_4$.

If $v_i \in M$, let $M= M \setminus\{v_i\} \cup \{C_i\}$, and the set of monitored edges does not decrease.
If $c_{i,j} \in M$, let $M= M \setminus\{c_{i,j}\} \cup \{C_i\}$, and the set of monitored edges does not decrease.
If $x_j \in M$, let $M= M \setminus\{x_j\} \cup \{c_{i,j}\}$ for some $i$ such that $x_j \in C_i$, and the set of monitored edges does not decrease.
If $c'_{i,j} \in M$, let $M= M \setminus\{c'_{i,j}\} \cup \{c_{i,j}\}$, and the set of monitored edges does not decrease.
If $w'_j \in M$, let $M= M \setminus\{w'_j\} \cup \{x_j\}$, and the set of monitored edges does not decrease.
Iterating this process, we finally obtain a distance-edge-monitoring set $M'$ of $G$ with $|M' \cap V_1|=1$,
$|M' \cap V'_1| \geq 1$, $|M' \cap V_3| \leq k$, $M' \cap (V_2 \cup V_4 \cup V_5 \cup V'_2 \cup V'_3) = \emptyset$.
Let $C= M' \cap V_3$. If $C$ is not a set cover of $H$, then there is $x_j \in X$ that is not covered by $C$. Recall that we assumed that each vertex of $X$
belongs to at least two edges of $S$. Then, according to Lemma~\ref{lemma:M(x)}, any edge $(c_{i,j},x_j)$ is not monitored by any of the vertices in $M' = (M' \cap V_1) \cup (M' \cap V'_1) \cup C$, a contradiction.
Hence, $C$ is a set cover of $H$ of size at most $k$.

\paragraph{(2)}
We now describe the second reduction, in which the obtained instance is bipartite and has diameter 6.
Let $(H,k)=((X,S),k)$ be an instance of \SCdec, where $X= \{x_1,x_2,\ldots,x_{|X|}\}$,
$S = \{C_1,C_2,\ldots,C_{|S|}\}$ and $C_i=\{c_{i,j} \mid x_j \in C_i\}$.
Construct the following instance $(G,k+2)=((V,E),k+2)$ of \EMdec,
where $V=V_1 \cup V_2 \cup \cdots \cup V_5 \cup V'_1 \cup V'_2 \cup V'_3 \cup V''$,
$E= E_1 \cup E_2 \cup E_3 \cup E_4 \cup E'_1 \cup E'_2 \cup E'_3 \cup E''$ and
\begin{eqnarray*}
V_1 & = & \{u_1,u_2,u_3,u_4\}, \\
V_2 & = & \{v_i \mid 1 \leq i \leq |S|\}, \\
V_3 & = & S, \\
V_4 & = & \{ c_{i,j} \mid 1 \leq i \leq |S|, 1 \leq j \leq |X|, x_j \in C_i\}, \\
V_5 & = & X, \\[6pt]
V'_1 & = & \{u'_1,u'_2,u'_3,u'_4\}, \\
V'_2 & = & \{v'_i \mid 1 \leq i \leq |S|\}, \\
V'_3 & = & \{w'_j \mid 1 \leq j \leq |X|\}, \\[6pt]
V''\! & = & \{u''_1\}, \\[6pt]
E_1 & = & \{(u_1,u_2), (u_2,u_3), (u_3,u_4), (u_4,u_1)\}, \\
E_2 & = & \{(u_1,v_i), (v_i,C_i) \mid 1 \leq i \leq |S|\}, \\
E_3 & = & \{(C_i, c_{i,j}) \mid 1 \leq i \leq |S|, 1 \leq j \leq |X|, x_j \in C_i\}, \\
E_4 & = & \{(c_{i,j},x_j) \mid 1 \leq i \leq |S|, 1 \leq j \leq |X|, x_j \in C_i\}, \\[6pt]
E'_1 & = & \{(u'_1,u'_2), (u'_2,u'_3), (u'_3,u'_4), (u'_4,u'_1)\}, \\
E'_2 & = & \{(u'_1,v'_i), (v'_i,C_i) \mid 1 \leq i \leq |S|\}, \\
E'_3 & = & \{(u'_1,w'_j), (w'_j,x_j) \mid 1 \leq j \leq |X|\}, \\[6pt]
E'' & = & \{(u_1,u''_1), (u''_1,u'_1)\}.
\end{eqnarray*}

The reduction is depicted in Figure~\ref{fig:reduction-bip}.

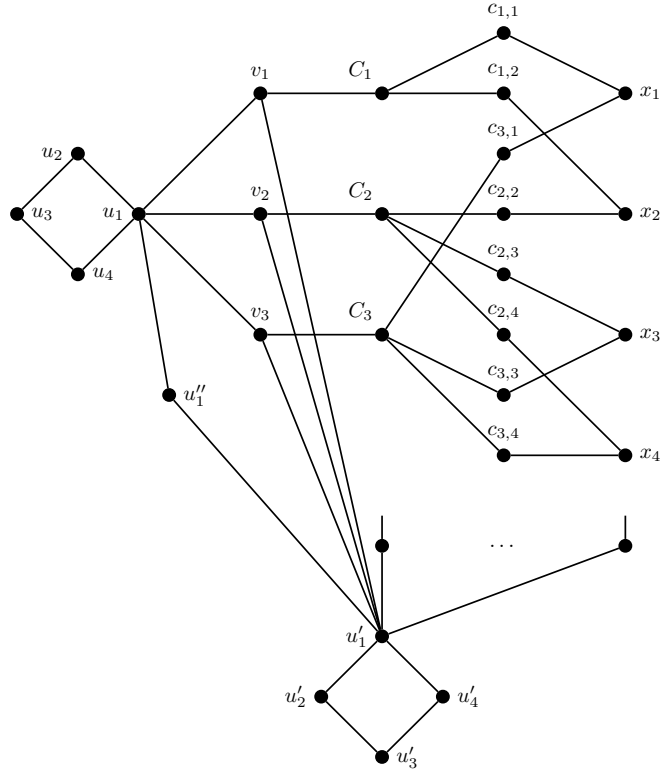
\begin{figure}[t]
\centering
\scalebox{0.8}{\begin{tikzpicture}[every loop/.style={},scale=1]
  \node[blackvertex,label={90:$v_1$}] (v1) at (2,2) {};
  \node[blackvertex,label={90:$v_2$}] (v2) at (2,0) {};
  \node[blackvertex,label={90:$v_3$}] (v3) at (2,-2) {};
  \node[blackvertex,label={100:$C_1$}] (C1) at (4,2) {};
  \node[blackvertex,label={100:$C_2$}] (C2) at (4,0) {};
  \node[blackvertex,label={100:$C_3$}] (C3) at (4,-2) {};
  \node[blackvertex,label={90:$c_{1,1}$}] (c11) at (6,3) {};
  \node[blackvertex,label={90:$c_{1,2}$}] (c12) at (6,2) {};
  \node[blackvertex,label={90:$c_{3,1}$}] (c31) at (6,1) {};
  \node[blackvertex,label={90:$c_{2,2}$}] (c22) at (6,0) {};
  \node[blackvertex,label={90:$c_{2,3}$}] (c23) at (6,-1) {};
  \node[blackvertex,label={90:$c_{2,4}$}] (c24) at (6,-2) {};
  \node[blackvertex,label={90:$c_{3,3}$}] (c33) at (6,-3) {};
  \node[blackvertex,label={90:$c_{3,4}$}] (c34) at (6,-4) {};
    \node[blackvertex,label={0:$x_1$}] (x1) at (8,2) {};
  \node[blackvertex,label={0:$x_2$}] (x2) at (8,0) {};
  \node[blackvertex,label={0:$x_3$}] (x3) at (8,-2) {};
  \node[blackvertex,label={0:$x_4$}] (x4) at (8,-4) {};
   \node[blackvertex,label={180:$u_1$}] (u1) at (0,0) {}; 
   \node[blackvertex,label={180:$u_2$}] (u2) at (-1,1) {}; 
   \node[blackvertex,label={0:$u_3$}] (u3) at (-2,0) {}; 
  \node[blackvertex,label={0:$u_4$}] (u4) at (-1,-1) {};
    \node[blackvertex,label={0:$u''_1$}] (u''1) at (0.5,-3) {};
   \node[blackvertex,label={180:$u'_1$}] (u'1) at (4,-7) {}; 
   \node[blackvertex,label={180:$u'_2$}] (u'2) at (3,-8) {}; 
   \node[blackvertex,label={0:$u'_3$}] (u'3) at (4,-9) {}; 
      \node[blackvertex,label={0:$u'_4$}] (u'4) at (5,-8) {}; 

   \node[blackvertex,label={}] (v'1) at (4,-5.5) {}; 
   \node[blackvertex,label={}] (v'3) at (8,-5.5) {}; 
   \path (6,-5.5) node {\ldots};

  \draw[thick] (u1) -- (v1) -- (C1);
  \draw[thick] (u1) -- (v2) -- (C2);
  \draw[thick] (u1) -- (v3) -- (C3);
  \draw[thick] (u4) -- (u1) -- (u2) -- (u3) -- (u4);
  \draw[thick] (u1) -- (u''1) -- (u'1);
  \draw[thick] (u'4) -- (u'1) -- (u'2) -- (u'3) -- (u'4);
  \draw[thick] (C1) -- (c11) -- (x1);
  \draw[thick] (C1) -- (c12) -- (x2);
  \draw[thick] (C2) -- (c22) -- (x2);
  \draw[thick] (C2) -- (c23) -- (x3);
  \draw[thick] (C2) -- (c24) -- (x4);
  \draw[thick] (C3) -- (c33) -- (x3);
  \draw[thick] (C3) -- (c34) -- (x4);
  \draw[thick] (C3) -- (c31) -- (x1);
  \draw[thick] (u'1) -- (v1);
  \draw[thick] (u'1) -- (v2);
  \draw[thick] (u'1) -- (v3);
  \draw[thick] (u'1) -- (v'1) -- ++(0,0.5);
  \draw[thick] (u'1) -- (v'3) -- ++(0,0.5);
  
  
\end{tikzpicture}}
\caption{Second reduction from \SCdec to \EMdec from the proof of Theorem~\ref{thm:approx-hardness} applied to the hypergraph $(\{x_1,x_2,x_3,x_4\},\{C_1=\{x_1,x_2\},C_2=\{x_2,x_3,x_4\},C_3=\{x_1,x_3,x_4\}\})$. Vertices and edges of $V'_i$ and $E'_i$ for $i=2,3$ are only suggested.}\label{fig:reduction-bip}
\end{figure}

Let $C$ be a set cover of $H$ of size $k$. Define $M = C \cup \{u_3,u'_3\}$. Then, by Lemma~\ref{lemma:M(x)},
$u_3$ monitors the edges $(u_2,u_3)$, $(u_3,u_4)$, $(u_1,u''_1)$, $(u'_2,u'_1)$, $(u'_4,u'_1)$  and all the edges in $E_2 \cup E_3$,
and $u'_3$ monitors the edges $(u'_2,u'_3)$, $(u'_3,u'_4)$, $(u''_1,u'_1)$, $(u_2,u_1)$, $(u_4,u_1)$ and all the edges
in $E'_2 \cup E'_3 \cup E'_4$.
If $e=(c_{i,j},x_j)$, there is $i'$ such that $x_j \in C_{i'}$ and $C_{i'} \in C$, and $e$ is monitored by $C_{i'}$
(either $i=i'$ and the only shortest path from $x_j$ to $C_i$ contains $e$, or $i\neq i'$ and the only shortest path from
$c_{i,j}$ to $C_{i'}$ contains $e$). Hence, $M$ is a distance-edge-monitoring set of $G$ of size at most $k+2$.

Let $M$ be a distance-edge-monitoring set of $G$ of size at most $k+2$.
In order to monitor the edge $(u_2,u_3)$, either $u_2\in M$ or $u_3 \in M$.
If $v_i \in M$, let $M= M \setminus\{v_i\} \cup \{C_i\}$, and the set of monitored edges does not decrease.
If $c_{i,j} \in M$, let $M= M \setminus\{c_{i,j}\} \cup \{C_i\}$, and the set of monitored edges does not decrease.
If $x_j \in M$, let $M= M \setminus\{x_j\} \cup \{c_{i,j}\}$ for some $i$ such that $x_j \in C_i$, and the set of monitored edges does not decrease.
In order to monitor the edge $(u'_2,u'_3)$, either $u'_2\in M$ or $u'_3 \in M$.
If $v'_i \in M$, let $M= M \setminus\{v'_i\} \cup \{C_i\}$, and the set of monitored edges does not decrease.
If $w'_j \in M$, let $M= M \setminus\{w'_j\} \cup \{x_j\}$, and the set of monitored edges does not decrease.
If $u''_1 \in M$, let $M= M \setminus\{u''_1\} \cup \{u_3\}$, and the set of monitored edges does not decrease.
Iterating this process, we finally obtain a distance-edge-monitoring set $M'$ of $G$ with $|M' \cap V_1| \geq 1$,
$|M' \cap V'_1| \geq 1$, $|M' \cap V_3| \leq k$, $M' \cap (V_2 \cup V_4 \cup V_5 \cup V'_2 \cup V'_3 \cup V'') = \emptyset$.
Let $C= M' \cap V_3$. If $C$ is not a set cover of $H$, then there is $x_j \in X$ that is not covered by $C$. Without loss of generality, assume that each vertex of $X$
belongs to at least two edges of $C$. Then, according to Lemma~\ref{lemma:M(x)}, any edge $(c_{i,j},x_j)$ is
not monitored by any of the vertices in $M' = (M' \cap V_1) \cup (M' \cap V'_1) \cup C$, a contradiction.
Hence, $C$ is a set cover of $H$ of size at most $k$.

\paragraph{(3)}
We now describe the third reduction, in which the obtained instance is bipartite and has maximum degree 3.
Let $(H,k)=((X,S),k)$ be an instance of \SCdec, where $X= \{x_1,x_2,\ldots,x_{|X|}\}$,
$S = \{C_1,C_2,\ldots,C_{|S|}\}$ and $C_i=\{c_{i,j} \mid x_j \in C_i\}$.
Let $D = \lceil\log_2 \max\{|X|,|S|\}\rceil$.
Let $B(r,\{r_1,r_2,\ldots,\allowbreak r_{\ell}\},d)$ denote a binary tree with root $r$ and $\ell$ leaves $r_1,r_2,\ldots,r_{\ell}$ at distance 
$d\geq \lceil\log_2 \ell\rceil$ from $r$.
Let $P(r,s,\ell)$ denote the path of length $\ell$ with endpoints $r$ and $s$.
Let $\mathcal C_{i} = \{ c_{i,j} \mid 1 \leq i \leq |S|, 1 \leq j \leq |X|, x_j \in C_i\}$.
Let $\mathcal X_j = \{ c_{i,j} \mid 1 \leq i \leq |S|, 1 \leq j \leq |X|, x_j \in C_i\}$.
Let $B = B(u,\{v_1,v_2,\ldots,v_{|S|}\},D)$.
For $1 \leq i \leq |S|, 1 \leq j \leq |X|$, let $P_i = P(v_i,C_i,D)$,
$B(C_i) = B(C_i,\mathcal C_{i},D)$, and
$B(x_j) = B(x_j,\mathcal X_j,D)$.
Construct the following instance $(G,k+1)=((V,E),k+1)$ of \EMdec,
where $V=V_1 \cup V_2 \cup \cdots \cup V_8$,
$E= E_1 \cup E_2 \cup \cdots \cup E_5$ and
\begin{eqnarray*}
V_1 & = & \{u_1,u_2,u_3,u_4\}, \\
V_2 & = & V(B), \\
V_3 & = & \cup_{1 \leq i \leq |S|} \, V(P_i), \\
V_4 & = & S, \\
V_5 & = & \cup_{1 \leq i \leq |S|, 1 \leq j \leq |X|, x_j \in C_i} \, V(B(C_i)), \\
V_6 & = & \cup_{1 \leq i \leq |S|, 1 \leq j \leq |X|, x_j \in C_i} \, \mathcal C_{i}, \\
V_7 & = & \cup_{1 \leq j \leq |X|} \, V(B(x_j)), \\
V_8 & = & X, \\[6pt]
E_1 & = & \{(u_1,u_2), (u_2,u_3), (u_3,u_4), (u_4,u_1), (u_1,u)\}, \\
E_2 & = & E(B), \\
E_3 & = & \cup_{1 \leq i \leq |S|} \, E(P_i), \\
E_4 & = & \cup_{1 \leq i \leq |S|, 1 \leq j \leq |X|, x_j \in C_i} \, E(B(C_i)), \\
E_5 & = & \cup_{1 \leq j \leq |X|} \, E(B(x_j)). 
\end{eqnarray*}

The reduction is depicted in Figure~\ref{fig:reduction-subcubic}.

\begin{figure}[!htpb]
\centering
\scalebox{0.8}{\begin{tikzpicture}[every loop/.style={},scale=1]
  \node[blackvertex,label={180:$u_1$}] (u1) at (-1,0) {};
  \node[blackvertex,label={90:$u$}] (u) at (0,0) {};
  \node[blackvertex] (b1) at (1,1) {};
  \node[blackvertex] (b2) at (1,-1) {};
  \node[blackvertex,label={90:$v_1$}] (v1) at (2,2) {};
  \node[blackvertex,label={90:$v_2$}] (v2) at (2,0) {};
  \node[blackvertex,label={90:$v_3$}] (v3) at (2,-2) {};
  \node[blackvertex] (p1) at (3,2) {};
  \node[blackvertex] (p2) at (3,0) {};
  \node[blackvertex] (p3) at (3,-2) {};
  \node[blackvertex,label={100:$C_1$}] (C1) at (4,2) {};
  \node[blackvertex,label={100:$C_2$}] (C2) at (4,0) {};
  \node[blackvertex,label={100:$C_3$}] (C3) at (4,-2) {};
  \node[blackvertex] (b11) at (5.5,3) {};
  \node[blackvertex] (b12) at (5.5,2) {};
  \node[blackvertex] (b21) at (5.5,0) {};
  \node[blackvertex] (b22) at (5.5,-1) {};
  \node[blackvertex] (b31) at (5.5,-2) {};
  \node[blackvertex] (b32) at (5.5,-3) {};
  \node[blackvertex,label={90:$c_{1,1}$}] (c11) at (7,3) {};
  \node[blackvertex,label={90:$c_{1,2}$}] (c12) at (7,2) {};
  \node[blackvertex,label={90:$c_{3,1}$}] (c31) at (7,1) {};
  \node[blackvertex,label={90:$c_{2,2}$}] (c22) at (7,0) {};
  \node[blackvertex,label={90:$c_{2,3}$}] (c23) at (7,-1) {};
  \node[blackvertex,label={90:$c_{2,4}$}] (c24) at (7,-2) {};
  \node[blackvertex,label={90:$c_{3,3}$}] (c33) at (7,-3) {};
  \node[blackvertex,label={90:$c_{3,4}$}] (c34) at (7,-4) {};
  \node[blackvertex] (bx11) at (8.5,3) {};
  \node[blackvertex] (bx21) at (8.5,2) {};
  \node[blackvertex] (bx12) at (8.5,1) {};
  \node[blackvertex] (bx22) at (8.5,0) {};
  \node[blackvertex] (bx31) at (8.5,-1) {};
  \node[blackvertex] (bx41) at (8.5,-2) {};
  \node[blackvertex] (bx32) at (8.5,-3) {};
  \node[blackvertex] (bx42) at (8.5,-4) {};
    \node[blackvertex,label={0:$x_1$}] (x1) at (10,2) {};
  \node[blackvertex,label={0:$x_2$}] (x2) at (10,0) {};
  \node[blackvertex,label={0:$x_3$}] (x3) at (10,-2) {};
  \node[blackvertex,label={0:$x_4$}] (x4) at (10,-4) {};
   \node[blackvertex,label={180:$u_2$}] (u2) at (-2,1) {}; 
   \node[blackvertex,label={0:$u_3$}] (u3) at (-3,0) {}; 
  \node[blackvertex,label={0:$u_4$}] (u4) at (-2,-1) {};

  \draw[thick] (u1) -- (u) (b1) -- (v1) -- (p1) -- (C1);
  \draw[thick] (u) -- (b1) -- (v2) -- (p2) -- (C2);
  \draw[thick] (u) -- (b2) -- (v3) -- (p3) -- (C3);
  \draw[thick] (u4) -- (u1) -- (u2) -- (u3) -- (u4) (u1) -- (u);
  \draw[thick] (C1) -- (b11) -- (c11) -- (bx11) -- (x1);
  \draw[thick] (C1) -- (b12) -- (c12) -- (bx12) -- (x2);
  \draw[thick] (C2) -- (b21) -- (c22) -- (bx22) -- (x2);
  \draw[thick] (b21) -- (c23) -- (bx31) -- (x3);
  \draw[thick] (C2) -- (b22) -- (c24) -- (bx41) -- (x4);
  \draw[thick] (C3) -- (b32) -- (c33) -- (bx32) -- (x3);
  \draw[thick] (b32) -- (c34) -- (bx42) -- (x4);
  \draw[thick] (C3) -- (b31) -- (c31) -- (bx21)-- (x1);

  
\end{tikzpicture}}
\caption{Third reduction from \SCdec to \EMdec from the proof of Theorem~\ref{thm:approx-hardness} applied to the hypergraph $(\{x_1,x_2,x_3,x_4\},\{C_1=\{x_1,x_2\},C_2=\{x_2,x_3,x_4\},C_3=\{x_1,x_3,x_4\}\})$, thus $D=\lceil\log_2 \max\{3,4\}\rceil=2$.}\label{fig:reduction-subcubic}
\end{figure}
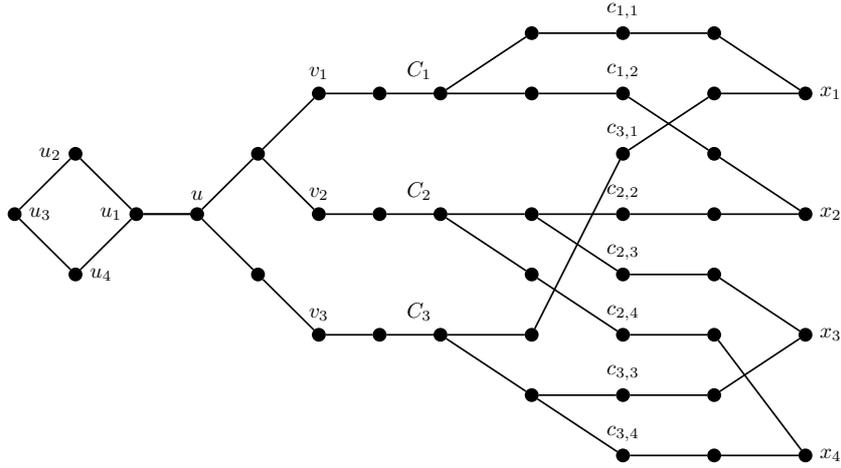

Let $C$ be a set cover of $H$ of size $k$. Define $M = C \cup \{u_2\}$. Then, by Lemma~\ref{lemma:M(x)},
$u_2$ monitors the edges $(u_1,u_2)$, $(u_2,u_3)$, $(u_1,u)$  and all the edges in $E_2 \cup E_3 \cup E_4$,
and any vertex in $C$ monitors the edges $(u_1,u_4)$, $(u_1,u_2)$.
If $e \in E(B(x_j))$, there is $i'$ such that $x_j \in C_{i'}$ and $C_{i'} \in C$, and $e$ is monitored by $C_{i'}$
(either $i=i'$ and the only shortest path from $x_j$ to $C_i$ contains $e$, or $i\neq i'$ and the only shortest path from
$c_{i,j}$ to $C_{i'}$ contains $e$). Hence, $M$ is a distance-edge-monitoring set of $G$ of size at most $k+1$.

Let $M$ be a distance-edge-monitoring set of $G$ of size at most $k+1$.
In order to monitor the edge $(u_2,u_3)$, either $u_2\in M$ or $u_3 \in M$.
If $v \in (M \cap V(B))$, descend to the leftmost leaf $v_i$ in $B$ with $v$ as its ancestor, and let $M= M \setminus\{v_i\} \cup \{C_i\}$.
Then, the set of monitored edges does not decrease.
If $v \in (M \cap V(P_i))$, let $M= M \setminus\{v_i\} \cup \{C_i\}$, and the set of monitored edges does not decrease.
If $v \in (M \cap V(B(C_i)))$, let $M= M \setminus\{v\} \cup \{C_i\}$, and the set of monitored edges does not decrease.
If $v \in (M \cap V(B(x_j)))$, let $M= M \setminus\{x_j\} \cup \{c_{i,j}\}$ for some $i$ such that $x_j \in C_i$,
and the set of monitored edges does not decrease.
Iterating this process, we finally obtain a distance-edge-monitoring set $M'$ of $G$ with
$|M' \cap V_1| \geq 1$, $|M' \cap V_4| \leq k$, $M' \cap (V_2 \cup V_3 \cup V_5 \cup V_6 \cup V_7 \cup V_8) = \emptyset$.
Let $C= M' \cap V_4$. If $C$ is not a set cover of $H$, then there is $x_j \in X$ that is not covered by $C$. Without loss of generality, assume that each vertex of $X$
belongs to at least two edges of $C$. Then, according to Lemma~\ref{lemma:M(x)}, any edge $(v,x_j)$ is
not monitored by any of the vertices in $M' = (M' \cap V_1) \cup C$, a contradiction.
Hence, $C$ is a set cover of $H$ of size at most $k$.
%
\end{proof}

\section{Concluding remarks and questions}\label{sec:conclu}

We have introduced a new graph parameter useful in the area of network monitoring. We conclude the paper with some remarks for future research directions.

\subsection{Structural graph parameters}

We have related the parameter $\EM$ to other standard graph parameters by giving lower and upper bounds. The diagram in Figure~\ref{fig:diagram} shows the relations between these parameters. It would be interesting to improve them. In particular, is it true that $\EM(G)\leq \FES(G)+1$? As we have seen, this bound would be tight. Moreover, is the bound $\EM(G)\geq\omega(G)/2$ from Theorem~\ref{thm:arboricity} tight?


The relation of the parameter $\EM$ to other structural graph parameters not considered here would also be of interest.
A \emph{feedback vertex set} of a graph $G$ is a set of vertices such that removing them from $G$ leaves a forest. The smallest size of a feedback vertex set of $G$ is denoted by $\FVS(G)$. It is not difficult to show that for any graph $G$, $\FVS(G)\leq\VC(G)$ and $\FVS(G)\leq\FES(G)$. However, we cannot hope to bound $\EM(G)$ by a function of $\FVS(G)$, as we did in Theorems~\ref{thm:VC} and~\ref{thm:fes} for $\VC(G)$ and $\FES(G)$. Indeed, consider the example of a path $P_a$ of order $a\geq 9$ to which we add a universal vertex. This graph $P_a\bowtie K_1$ has radius at least~$4$, thus, by Theorem~\ref{thm:VC-universal-vertex}, we have $\EM(P_a\bowtie K_1)=\VC(P_a)=\lfloor a/2\rfloor$, while we have $\FVS(P_a\bowtie K_1)=1$. Similarly, the pathwidth of $P_a\bowtie K_1$ is $2$. Nevertheless, we do not know whether we can upper-bound $\EM(G)$ by a function of the bandwidth or the treedepth of $G$.

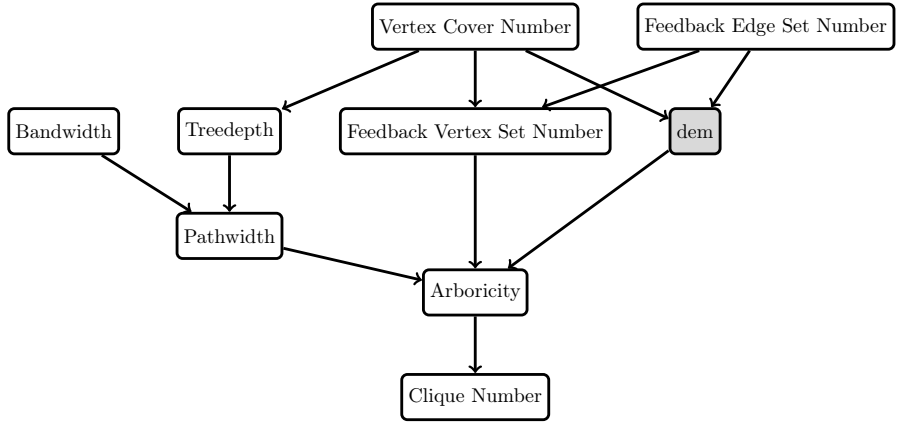
\begin{figure}[ht!]
\centering
\tikzstyle{mybox}=[line width=0.5mm,rectangle, minimum height=.8cm,fill=white!70,rounded corners=1mm,draw]
\tikzstyle{myedge}=[line width=0.5mm,<-]
\scalebox{0.75}{\begin{tikzpicture}[node distance=10mm]
 		\node[mybox] (vc)  {Vertex Cover Number};
  		\node[mybox] (fes) [right = of vc] {Feedback Edge Set Number};
 		\node[mybox] (fvs) [below = of vc] {Feedback Vertex Set Number} edge[myedge] (vc) edge[myedge] (fes);
 		\node[mybox] (td) [left = of fvs] {Treedepth} edge[myedge] (vc);
 		\node[mybox] (bw) [left = of td] {Bandwidth};
 		\node[mybox] (pw) [below = of td] {Pathwidth} edge[myedge] (td) edge[myedge] (bw);
 		\node[mybox] (dem) [right = of fvs,fill=black!15] {$\EM$} edge[myedge] (vc) edge[myedge] (fes);
 		\node[mybox] (arb) [below = of fvs,yshift=-10mm] {Arboricity} edge[myedge] (fvs) edge[myedge] (pw) edge[myedge] (dem);
 		\node[mybox] (c) [below = of arb] {Clique Number} edge[myedge] (arb);

\end{tikzpicture}}
\caption{Relations between some structural parameters and $\EM$. Arcs between parameters indicate that the bottom parameter is upper-bounded by a function of the top parameter.}
\label{fig:diagram}
\end{figure}

\subsection{Further complexity questions}

It would also be interesting to determine graph classes where \EMdec has a polynomial-time (or parameterized) exact or constant-factor approximation algorithm. For example, what happens for planar graphs, chordal graphs or interval graphs?

\subsection{Strenghtening the definition}

It is also reasonable to require to be able to simultaneously detect multiple edge failures in a network. In that case, we could propose the following strengthening of distance-edge-monitoring sets, based on extending Proposition~\ref{prop:location}. For a graph $G$ and a set $E$ of edges of $G$, we denote by $G-E$ the graph obtained by removing all edges of $E$ from $G$.

\begin{definition}
  For a set $E$ of edges and a set $M$ of vertices of a graph $G$, we let $P(M,E)$ be the set of pairs $(x,y)$ with $x$ a vertex of $M$ and $y$ a vertex of $V(G)$ such that $d_G(x,y)\neq d_{G-E}(x,y)$.

  For a graph $G$, a distance-edge-monitoring set $M$ of $G$ is called a \emph{strong distance-edge-monitoring set} of $G$ if for any two distinct subsets $E_1$ and $E_2$ of edges of $G$, we have $P(M,E_1)\neq P(M,E_2)$. We denote by $\SEM(G)$ the smallest size of a strong distance-edge-monitoring set of $G$.
\end{definition}

However, it turns out that the concept of strong distance edge-monitoring sets is in fact equivalent to the one of a vertex cover!

\begin{proposition}
A set $M$ of vertices of a graph $G$ is strong distance-edge-monitoring if and only if it is a vertex cover of $G$.
\end{proposition}
\begin{proof}
  First, assume that $M$ is a vertex cover. For an edge $uv$ (assume without loss of generality that $u\in M$), we have $(u,v)\in P(M,E)$ if and only if $uv\in E$. For two distinct sets $E_1,E_2$ of edges of $G$, we have an edge $e=uv$ in the symmetric difference $E_1\ominus E_2$ and thus one of the pairs $(u,v)$ or $(v,u)$ belongs to exactly one of $P(M,E_1)$ and $P(M,E_2)$. This implies that $M$ is a strong distance-edge-monitoring set.

  Conversely, let $M$ be a strong distance-edge-monitoring set of $G$. Assume for a contradiction that it is not a vertex cover of $G$, so, there is an edge $uv$ with $\{u,v\}\cap M=\emptyset$. Let $E$ be the set of edges that are incident with $u$ or with $v$ (or both), and $E'=E\setminus\{uv\}$. Then, for any two vertices $x,y$ with $x\notin\{u,v\}$, we have $d_{G-E}(x,y)=d_{G-E'}(x,y)$. In particular, this is true when $x\in M$. Thus, we have $P(M,E)=P(M,E')$, a contradiction.
\end{proof}

Nevertheless, it could remain interesting to study an intermediate concept, where we wish to monitor all sets of edges of size at most a given constant. This has been done for example in the context of the metric dimension in~\cite{MDsets}.







\end{document}